\documentclass[12pt]{article}
\usepackage{amsmath}
\usepackage{graphicx}
\usepackage{enumerate}
\usepackage{natbib}
\usepackage{url} % not crucial - just used below for the URL 

\newcommand{\blind}{1}

\usepackage{bbding}
\usepackage{pifont}% http://ctan.org/pkg/pifont
\usepackage{multirow}
\usepackage{amsfonts}
\usepackage{bbding}
\usepackage{float}
\usepackage{amssymb, latexsym}
\usepackage{amsthm,color}
\usepackage{latexsym}
\usepackage{bm}

\newtheorem{lem}{Lemma}

\newtheorem{thm}{Theorem}
\newtheorem{cor}{Corollary}

\newtheorem{rem}{Remark}
\newtheorem{defi}{Definition}
\numberwithin{equation}{section}

\DeclareMathOperator*{\argmin}{arg\,min}
\DeclareMathOperator*{\argmax}{arg\,max}
\newcommand{\xmark}{\ding{55}}%

\newcommand{\bA}{\boldsymbol{A}}
\newcommand{\br}{\boldsymbol{r}}
\newcommand{\bx}{\bm{x}}

\newcommand{\by}{\bm{y}}
\newcommand{\bu}{\bm{u}}

\newcommand{\bz}{\bm{z}}

\newcommand{\bX}{\boldsymbol{X}}

\newcommand{\bQ}{\boldsymbol{Q}}

\newcommand{\bw}{\boldsymbol{w}}

\newcommand{\bM}{\boldsymbol{M}}
\newcommand{\bP}{\boldsymbol{P}}

\newcommand{\bbeta}{\boldsymbol{\beta}}

\newcommand{\btheta}{\boldsymbol{\theta}}

\newcommand{\bgamma}{\boldsymbol{\gamma}}

\newcommand{\balpha}{\boldsymbol{\alpha}}

\newcommand{\bzero}{\boldsymbol{0}}
\newcommand{\var}{\mathrm{var}}
\newcommand{\pr}{\mathrm{pr}}
\newcommand{\intd}{\mathrm{d}}
\newcommand{\logit}{\mathrm{logit\,}}

\newcommand{\Cov}{\mathrm{cov}}
\newcommand{\T}{\mathrm{T}}

\newcommand{\norm}[1]{\Vert#1\Vert}

\begin{document}

%\bibliographystyle{natbib}

%%%%%%%%%%%%%%%%%%%%%%%%%%%%%%%%%%%%%%%%%%%%%%%%%%%%%%%%%%%%%%%%%%%%%%%%%%%%%%

\if1\blind
{
  \title{\bf Functional Calibration under Non-Probability Survey Sampling}
  \author{Zhonglei Wang \\
    Wang Yanan Institute for Studies in Economics and School of Economics, \\
    Xiamen University\\
    and \\
    Xiaojun Mao\thanks{
    Zhonglei Wang and Xiaojun Mao contribute equally.}\hspace{.2cm} \\
    School of Mathematical Sciences, Shanghai Jiao Tong University\\
    and\\
    Jae Kwang Kim\\
    Department of Statistics, Iowa State University}
    \date{}
  \maketitle
} \fi

\if0\blind
{
\title{\bf Functional Calibration under Non-Probability Survey Sampling}
\author{}
\date{}
\maketitle
%   \bigskip
%   \bigskip
%   \bigskip
%   \begin{center}
%     {\LARGE\bf Functional Calibration under Non-Probability Survey Sampling}
%     \end{center}
%   \medskip
} \fi

\bigskip
\begin{abstract}
 Non-probability sampling is prevailing in survey sampling, but ignoring its selection bias leads to erroneous inferences. We offer a unified nonparametric calibration method to estimate the sampling weights for a non-probability sample  by calibrating functions of auxiliary variables in a reproducing kernel Hilbert space. The consistency and the limiting distribution of the proposed estimator are established, and the corresponding variance estimator is also investigated. Compared with existing works, the proposed method is more robust since no parametric assumption is made for the selection mechanism of the non-probability sample. Numerical results demonstrate that the proposed method outperforms its competitors, especially when the model is misspecified. The proposed method is applied to analyze the average total cholesterol of Korean citizens based on a non-probability sample from the National Health Insurance Sharing Service and a reference probability sample from the Korea National Health and Nutrition Examination Survey.
\end{abstract}

\noindent%
{\it Keywords:}   Data integration;  Missing at random; Nonparametric weighting; Reproducing kernel Hilbert space.
\vfill

\def\spacingset#1{\renewcommand{\baselinestretch}%
{#1}\small\normalsize} \spacingset{1}
\spacingset{1.9} % DON'T change the spacing!

\section{Introduction}
Probability sampling serves as a golden standard to estimate finite population parameters in social science \citep{elliott2017inference,haziza2017construction}, but low response rates and inevitable dropouts have made it ``tarnished gold'' recently \citep{keiding2016perils}. Moreover, it is costly and time-consuming to conduct probability sampling, so it is only feasible for well-funded and socially important surveys \citep{baker2013summary,o2014generalizing}.   On the other hand, due to its feasibility and low cost, non-probability sampling, especially web surveys, has become increasingly popular \citep{couper2000web,couper2008web,dever2008internet,tourangeau2013science,dever2014estimation}. Nonetheless, a non-probability sample is rarely  representative of the target population because of its unknown selection mechanism. If such a selection mechanism is not properly incorporated, it may lead to erroneous inferences.  Therefore, adjusting the selection bias for a non-probability sample has become a hot research topic in survey sampling. 

If an additional reference probability sample is available, there are primarily two techniques to adjust the selection bias for a non-probability sample. One method involves combining the non-probability sample and the reference probability sample to calculate propensity scores \citep{rosenbaum1983central}.   Under a parametric assumption for the response model, \citet{lee2006propensity} developed a quasi-randomization method to estimate the propensity scores for the pooled sample. The pooled sample is then divided into groups based on the 
estimated propensity scores, and modified sampling weights for the non-probability sample are calculated for each group.  \citet{lee2009estimation} generalized the quasi-randomization method \citep{lee2006propensity} by an additional calibration adjustment \citep{deville1992calibration} for the case when marginal population totals of  auxiliaries are available. \citet{valliant2011estimating} compared several propensity-score-based estimators and concluded that sampling weights of the reference probability sample should be incorporated when estimating the propensity scores. Also see \citet{rivers2007sampling}, \citet{bethlehem2010selection}, \citet{brick2015compositional}, \citet{elliott2017inference} and the references within for more details. The second approach uses calibration   to adjust the selection bias of a non-probability sample. \citet{kim2019sampling} estimated the ``importance weights'' for a non-probability sample based on the Kullback-Leibler (KL) divergence, and they only assumed the availability of the marginal population totals for auxiliaries. \citet{chen2020doubly} proposed to estimate the parameters in the propensity score model based on a calibration constraint \citep{wu2001model,beaumont2005calibrated}, and a reference probability sample is used to estimate the population totals of a specific estimating function.

Existing works  assume a  parametric model either for the propensity scores \citep{elliott2007use,chen2020doubly} or for the sampling weights \citep{kim2019sampling}, so they suffer from model misspecification \citep{robins1994estimation,Han2013}. In this paper, we present a nonparametric method based on functional calibration in a reproducing kernel Hilbert space \citep[RKHS]{wahba1990spline} for adjusting 
the selection bias for a non-probability sample. Specifically, uniformly calibrating functions in an RKHS is utilized to get the estimated sampling weights of a non-probability sample, and the reference probability sample is used to estimate the associated population totals as \citet{chen2020doubly}.  In addition, we propose to use the KL divergence as a penalty to avoid  overfitting. 
 Under some regularity conditions, asymptotic properties of the proposed method are investigated, and numerical results demonstrate the advantages of the proposed method over its alternatives.

The proposed method differs from existing ones in the following aspects. Some existing methods \citep{valliant2011estimating,chen2020doubly} estimate propensity scores for a non-probability sample, so the corresponding estimator is inefficient if estimated propensity scores are close to zero. To avoid such inefficiency, we propose directly estimating the sampling weights and applying penalties  as well. Unlike \citet{chen2020doubly}, we do not make any parametric assumption for calibration.  Rather than that, the sampling weights of a non-probability sample are calculated via uniform  minimization of a calibration-based objective function over an RKHS. Thus, the proposed method outperforms the method of \citet{chen2020doubly} and other existing ones in terms of robustness.  {Since the proposed method uniformly calibrates functions in an RKHS, it is essentially a multitask-oriented learning  in the sense that the estimated sampling weights can be used to estimate several different parameters from  the non-probability sample. This property is appealing especially when the non-probability sample consists of  many survey questions.} Up to our knowledge, we are the first to  use  uniform calibration to  adjust the selection bias for a non-probability sample.

Although the proposed method is motivated by  \citet{wong2018kernel}, instead of assuming the auxiliaries to be available for every element in a finite population, we consider the setup when  a reference probability sample is used to estimate the population totals for functions in the RKHS. It is worth pointing out that assuming the availability of auxiliary information for each element in a finite population is generally unrealistic under survey sampling. Different from \citet{wong2018kernel}, we propose  a penalty based on a nonparametric density ratio model, and numerical results indicate that the proposed method is more efficient than theirs in terms of computation and estimation; see Section~\ref{sec: simu} for details. Besides, since the sampling indicators are no longer independent for the reference probability sample under rejective sampling, the  theoretical results from \citet{wong2018kernel} are not applicable to our method, and  we adopt a different empirical process technique instead. Even though empirical processes have been studied under survey sampling,  most of them assumed that the sample size is of the same order with the population size asymptotically \citep{breslow2007,conti2014estimation,bertail2017empirical,han2021}, but it is rarely the case for a reference probability sample in practice due to the limited budget.  \citet{Boistard2017} relaxed that stringent condition on the sample size, but they focused on single-stage sampling designs. 
In this paper, theoretical properties of the proposed method are investigated without assuming that the sizes of the reference probability sample and the finite population are of the same asymptotic order, and the proposed method applies as long as the reference probability sample is generated by a rejective sampling design, not limited to single-stage sampling.

The remaining of this paper is organized as follows. The motivation of the proposed method is introduced in Section~\ref{sec: basic setup}. The proposed method is presented in Section~\ref{sec: prop}, and its theoretical properties are investigated in Section~\ref{sec: asymptot}. Simulation studies are demonstrated in Section~\ref{sec: simu}. The proposed method is applied to estimate the average total cholesterol of the Korean citizens in Section~\ref{sec: application}. Concluding remarks are provided in Section~\ref{sec: conclusion}.

\section{Motivation}\label{sec: basic setup}

\subsection{Basic setup} 
To introduce the idea of uniform calibration,  assume  that  the finite population $\mathcal{F}_N = \{(\bx_i,y_i):i=1,2, \ldots, N\}$ is a random sample of size $N$ from a super-population model,
 	 	\begin{equation}
 	 		y_i = m(\bx_i) + \epsilon_i\quad (i=1,\ldots,N),\label{eq: popu model}
 	 	\end{equation}
  	where $\bx_i\in\mathcal{X}$, $\mathcal{X}\subset\mathbb{R}^d$ is a $d$-dimensional compact set, $y_i$ is the response of interest, $m(\bx_i)=E(y_i\mid\bx_i)$ is a smooth function \citep{wahba1990spline},  $\epsilon_i$ is independent with $\bx_i$, $E(\epsilon_i)=0$ and $E(\epsilon_i^2)=\sigma_i^2$, and $\sigma_1,\ldots,\sigma_N$ are positive constants with respect to the super-population model. We adopt a design-based framework and assume that the finite population $\mathcal{F}_N$ is fixed once it is generated; see Part~I of \citet{sarndal2003model}, Chapter~1 of \citet{fuller2009} and \citet{chen2020doubly} for details. The parameter of interest is the population mean $\bar{Y}_N = N^{-1}\sum_{i=1}^Ny_i$. For simplicity, assume that  the population size $N$ is known.

Let $A$ be a non-probability sample with observations on both the auxiliary vector and the response of interest, and $B$ be a reference probability sample with information on the auxiliary vector only. That is, both $\{(\bx_i,y_i):i\in A\}$ and $\{(\bx_i,\pi_{B,i}):i\in B\}$ are available, where $\pi_{B,i}$ is the first-order inclusion probability of the $i$-th element with respect to the probability sample $B$. Since the selection mechanism of the non-probability sample $A$ is unknown, this sample may  not  represent the finite population. Table \ref{table1} shows the general data structure of the two samples.  How to adjust the selection bias of the non-probability sample $A$ using the auxiliary information available from the probability sample $B$ is an important practical problem in survey sampling.

\begin{table}
\centering
\caption{Data structure of the two  samples. ``$\bX$'' denotes the auxiliary vector, and ``$Y$'' denotes the response of interest. ``\Checkmark'' is used if the information is available and ``\xmark'' otherwise.}
\label{table:1}
\begin{center}
\begin{tabular}{ccccc}
\hline
$\;\;\;$ Sample $\;\;\;$ & $\;\;\;$ Type $\;\;\;$ & $\;\;\;\;$ $\bX$ $\;\;\;\;$ & $\;\;\;\;$ $Y$ $\;\;\;\;$ & $\;$ Representativeness $\;$ \\
\hline
$A$ & Non-probability Sample & \Checkmark & \Checkmark & No \\
$B$ & Probability Sample & \Checkmark & \xmark & Yes \\
\hline 
\end{tabular}
\end{center}
\label{table1} 
\end{table}

A special case is  when the probability  sample $B$ is a census, and it was investigated by \citet{wong2018kernel}. However, a census is usually hard to obtain in practice. Thus, we focus on a more general case assuming that the probability sample $B$ is generated by a rejective sampling design. Under rejective sampling,  a sample is only acceptable  if a certain criterion is satisfied; see \citet[Section~1.2.6]{fuller2011sampling} and \citet{fuller2009} for details. Besides, the corresponding sampling indicators $\{\delta_{B,i}:i=1,\ldots,N\}$ are negatively associated \citep{bertail2016sharp}, where $\delta_{B,i}=1$ if $i\in B$ and 0 otherwise.
\subsection{Uniform calibration} 

To motivate the proposed method, we make a stronger assumption for (\ref{eq: popu model}) that there exists a positive constant $\sigma_0$ such that $\sigma_i=\sigma_0$ for $i=1,\ldots,N$. We consider an estimator of the form $\hat{Y}=N^{-1}\sum_{i\in A}{\omega}_iy_i$, where $\{{\omega}_i:i\in A\}$ are a set of weights to be determined. Under the super-population model  (\ref{eq: popu model}), we have
\begin{eqnarray}
    \hat{Y} - \bar{Y}_N &=& N^{-1} \left\{ \sum_{i \in A} {\omega}_i m(\bx_i) - \sum_{i =1}^N  m(\bx_i)  \right\} +   N^{-1}\left\{ \sum_{i \in A} {\omega}_i \epsilon_i  - \sum_{i=1}^N  \epsilon_i   \right\}\notag \\ 
    &:=& C+D.\label{eq: mot 1}
\end{eqnarray}
Thus, the weights $\{{\omega}_i:i\in A\}$ are optimal if  $Q=C^2+E\{D^2\}$ is minimized, where the expectation is taken with respect to the super-population model (\ref{eq: popu model}) conditional on the non-probability sample $A$. 
If the true mean function satisfied
\begin{equation} 
m(\bx) \in \mbox{span}\{ b_1(\bx),\ldots,b_L(\bx)\} \equiv \mathcal{H}_0 
\label{assume1}
\end{equation}  
for some basis functions $b_1(\bx),\ldots,b_L(\bx)$, and if $\{\bx_1,\ldots,\bx_N\}$ were available, then we could obtain $C^2=0$ by imposing 
\begin{equation} 
 \sum_{i \in A} \omega_i \left[ b_1(\bx_i), \ldots, b_L(\bx_i) \right] = \sum_{i=1}^N \left[ b_1(\bx_i), \ldots, b_L(\bx_i) \right] 
\label{calib2}
\end{equation} 
as the calibration equation for determining $\{\omega_i:i\in A\}$. Thus, the optimal weights are those minimizing 
$$ \sigma_0^{-2} E\{D^2 \} = \sum_{i \in A} \left( \omega_i-1 \right)^2 + \mbox{const} 
$$
subject to (\ref{calib2}). Assumption (\ref{assume1}) can be restrictive as the mean function is linear in the basis functions. However, we can still use the basis functions in the calibration equation to provide   nonparametric calibration estimates 
by allowing that the basis functions give a uniform approximation of the nonlinear function $m(\bx)$  with increasing dimension $L$.   For examples of nonparametric calibration estimation, \cite{montanari2005}  used a single-layer neural network model and  \cite{breidt2005model} considered a  penalized spline  model.

In our setup, instead of observing $\{\bx_i:i=1,\ldots,N\}$, however, only a reference probability sample $B$ is available. 
Since $\sum_{i \in B} \pi_{B,i}^{-1}  \left[ b_1(\bx_i), \ldots, b_L(\bx_i) \right]$ is design-unbiased for $\sum_{i=1}^N [ b_1(\bx_i), \ldots, b_L(\bx_i) ]$ \citep{horvitz1952generalization}, rather than  (\ref{calib2}), we impose
\begin{equation} 
 \sum_{i \in A} \omega_i \left[ b_1(\bx_i), \ldots, b_L(\bx_i) \right] = \sum_{i \in B} \pi_{B,i}^{-1}  \left[ b_1(\bx_i), \ldots, b_L(\bx_i) \right] 
\label{calib3}
\end{equation} 
as the calibration equation. Recall that the calibration (\ref{calib3}) is justified under (\ref{assume1}). Now, if assumption (\ref{assume1}) does not hold, then we  may intuitively consider minimizing 
\begin{equation}
 Q =  \sup_{ u \in \mathcal{H}} \left\{ \sum_{i \in A} {\omega}_i u(\bx_i) - \sum_{i \in B} \pi_{B,i}^{-1} u(\bx_i)  \right\}^2 + \sigma_0^2 \sum_{i\in A}({\omega}_i-1)^2 \label{eq: I II PRIME}
\end{equation}
directly for some function space $\mathcal{H}$. The first term of (\ref{eq: I II PRIME}) achieves the approximate uniform   calibration and the second term achieves the weight stabilization.   We assume that the function space $\mathcal{H}$ is an RKHS, so that we can construct certain basis functions in $\mathcal{H}$ from the sample; see Section~\ref{ss RKH} of the Supplementary Material for a brief introduction to an RKHS. 
In the next section, we also propose a different penalty term to stabilize the estimated weights, {and the advantage of the new penalty term is shown in Section~\ref{sec: simu}.}
\subsection{Assumptions} 
Before closing this section, we make the following assumptions  for the non-probability sample $A$: 
 \begin{enumerate}
 \renewcommand{\labelenumi}{A\arabic{enumi}.}
     \item The sampling indicators $\{\delta_{A,i}:i=1,\ldots,N\}$ are mutually independent, where $\delta_{A,i}=1$ if $i\in A$ and 0 otherwise. \label{ass: A ind}
    \item The sampling indicator $\delta_{A,i}$  is independent with  the response of interest $y_i$ given $\bx_i$. That is, 
$\pr(\delta_{A,i}=1\mid \bx_i,y_i)=\pi_A(\bx_i).$\label{ass: MAR}
 \end{enumerate}
 
The independence assumption in Assumption~A\ref{ass: A ind} is widely adopted for non-probability sampling \citep{keiding2016perils,chen2020doubly}; also see Section~17.2 of \citet{wu2021} for details.
The non-informative   assumption \citep{pfeffermann1993} in Assumption~A\ref{ass: MAR} is also commonly presumed for observational studies with  a sample similar as the non-probability one; see \citet{rosenbaum1983central} for details.

By Assumption~\ref{ass: MAR} and the Bayes formula, we have
 \begin{equation} 
 \pi_A( \bx_i) =  \frac{\pi_1 f( \bx_i \mid  \delta_{A,i}=1)  }{
 \pi_1 f ( \bx_i \mid \delta_{A,i}=1) + \pi_0 f (\bx_{i} \mid \delta_{A,i}=0)  }
 = \frac{\pi_1 f_1( \bx_i )  }{
 \pi_1 f_1 ( \bx_i ) + \pi_0 f_0 (\bx_i)  },
 \label{11-3-4}
 \end{equation} 
 where $\pi_1 = \int \pi_{A}(\bx)\intd \bx$, and $\pi_0 =1-\pi_1$, and $f_1(\bx_i)$ and $f_0(\bx_i)$ are the conditional probability densities of  $\bx_i$ given $\delta_{A,i}=1$ and $\delta_{A,i}=0$  with respect to a certain dominating measure $\mu$, respectively. For simplicity, we assume  the dominating measure $\mu$ to be the Lebesgue measure in this paper.  Writing $\omega^\star (\bx)= \{ \pi_A( \bx)\}^{-1}$, by (\ref{11-3-4}), we can obtain 
 \begin{equation}\label{eq: wr} \omega^\star (\bx)  = 1+ \frac{\pi_0}{ \pi_1 }r^{\star} (\bx) ,
 \end{equation}
 where 
 $r^{\star}( \bx) = f_0 ( \bx )/ f_1 ( \bx) $.  
That is, there is a one-to-one correspondence between  $ \omega_i^\star =\omega^\star( \bx_i)$ and  $r^\star_i=r^{\star} (\bx_i)$ for $i\in A$. In this paper, we propose to estimate $r^\star_i$ instead of $\omega_i^\star$, and its advantage is discussed in Remark~\ref{rem: KL diver} of the next section.

To regulate the selection probabilities associated with the non-probability sample $A$, we make the following assumption.
\begin{enumerate}
 \renewcommand{\labelenumi}{A\arabic{enumi}.}
 \setcounter{enumi}{2}
        \item There exist two positive constants $0<C_{A,1}<C_{A,2}<1$, such that  $C_{A,1}\leq \pi_A(\bx)\leq C_{A,2}$ for $\bx\in\mathcal{X}$. \label{ass:piB}
\end{enumerate}
The assumption $\pi_{A}(\bx)>C_{A,1}$ is slightly stronger than  assuming $\pi_A(\bx)>0$ for $\bx\in\mathcal{X}$; see Section~17.2 of \citet{wu2021} and Assumption~A2 of \citet{chen2020doubly} for comparison. However, such an assumption is required to derive the asymptotic properties of the proposed method; see the proof of Lemma~\ref{lemma: R S2} in Section~\ref{supp: proof of lemma R S2} of the Supplementary Material for details. Besides, Assumption~A\ref{ass:piB} implies that $n_B$ asymptotically has the same order as the population size $N$, and it makes sense in practice since a non-probability sample usually corresponds to  a big data source. The condition $\pi_A(\bx)<C_{A,2}$ for $\bx\in\mathcal{X}$ guarantees that the corresponding density ratio function $r^{\star}(\bx)$ is positive, so that the loss function (\ref{eq: popu level}) in the next section is valid for $r^{\star}(\bx)$. Specifically, by  (\ref{11-3-4}), we have 
\begin{equation}\label{eq: R x}
    r^\star(\bx) = \frac{f_0(\bx)}{f_1(\bx)} = \frac{\pi_1\{1-\pi_A(\bx)\}}{\pi_0\pi_A(\bx)},
\end{equation}
and by Assumption~A\ref{ass:piB}, we conclude that there exists two positive  constant $0<C_{r,1}< C_{r,2}$ depending only on $C_{A,1}$ and $C_{A,2}$, such that for $\bx\in\mathcal{X}$, we have                                 
\begin{equation}\label{eq: range r}
 C_{r,1}\leq r^\star(\bx)\leq C_{r,2}.
\end{equation}

\section{Proposed method}\label{sec: prop}
In Section 2, 
we have seen that the propensity score estimation problem reduces to the density ratio estimation problem. Density ratio estimation (DRE), the problem of estimating the ratio of two density functions for two different populations, is a fundamental problem in machine learning \citep{sugiyama2012}.
By partitioning the sample into two groups based on the response status, we can apply the DRE method and thus obtain the inverse propensity scores.
One important method of DRE is so called the maximum entropy method, which minimizes the KL divergence (or negative entropy)
subject to a normalization constraint \citep{nguyen2010}.

Applying the maximum entropy method of  \citet{nguyen2010}, the density ratio function $r^{\star}(\bx)$ can be understood as the maximizer of 
 \begin{equation}\label{eq: popu level}
     Q (r) =  \int  r\log \left( r \right) f_{1}  d \mu  - \int r f_1 d \mu + 1,
\end{equation} 
where $r=r(\bx)>0$ for $\bx\in\mathcal{X}$. 
A sample version of (\ref{eq: popu level}) is 
 \begin{equation} 
{Q}_A(\bgamma) =  \frac{1}{n_A} \sum_{i=1}^{N}   \delta_{A,i}   r_i  \left\{  \log ( r_i) -1 \right\}+1
\label{qq0},
\end{equation}  
where $\bgamma = (r_1,\ldots,r_N)^{\T}$, $r_i=r(\bx_i)$ if $i\in A$ and $r_i=0$ otherwise; see  (7.26) of \citet{kim2013statistical} for details. 

We propose to estimate $\{r_i^\star:i\in A\}$  by uniformly calibrating functions in an RKHS $\mathcal{H}$. Consider 
\begin{equation}
	\hat{\bgamma} = \argmin_{\xi_1\leq r_i\leq  \xi_2}\left[\sup_{u\in{\mathcal{H}}}\left\{\frac{S(\bgamma,u)}{\lVert u\rVert_2^2}-\lambda_1\frac{\lVert u\rVert_{\mathcal{H}}^2}{\lVert u\rVert_2^2}\right\}- \lambda_2 Q_A (\bgamma) \right],\label{eq: object2}
\end{equation}
where  $\xi_1\leq\xi_2$ are  predetermined numbers, $\lVert u\rVert_{\mathcal{H}}$ is the norm associated with the RKHS $\mathcal{H}$,  $\lVert u\rVert_2^2 = n^{-1}\sum_{i=1}^N(\delta_{A,i}+\delta_{B,i})u(\bx_i)^2$, $n=n_A+n_B$, $\lambda_1>0$ and $\lambda_2>0$ are two tuning parameters, and
\begin{eqnarray}
	S(\bgamma,u) &=& \left[ N^{-1}\sum_{i=1}^N\delta_{A,i} \left\{ 1+ \left( \frac{N}{n_A} -1\right)r_i \right\}u(\bx_i) - N^{-1}\sum_{i=1}^N\delta_{B,i} \pi_{B,i}^{-1}u(\bx_i)\right]^2. \notag \\
\label{eq:S2}
\end{eqnarray} 
In the optimization problem (\ref{eq: object2}), we should choose a sufficiently small $\xi_1$ and  a sufficiently large $\xi_2$ to guarantee $\xi_1\leq C_{r,1}<C_{r,2}\leq \xi_2$ by (\ref{eq: range r}). In practice, we can set $\xi_1 = 10^{-8}$ and $\xi_2 = 10^8$, for example. 
Since $\pi_0\pi_1^{-1}$ is generally unavailable, we replace it by $Nn_A^{-1}-1$ in (\ref{eq:S2}).
Different from \citet{wong2018kernel}, we assume an upper bound for $\{r_i:i=1,\ldots,N\}$ in the optimization problem (\ref{eq: object2}), and such an assumption is used to guarantee the convergence rate of $S(\hat{\bgamma},u)$ for $u\in\mathcal{H}$; see (\ref{eq: rate of s gamma hat}) in the Supplementary Material for details.  For simplicity, we implicitly assume that the auxiliary vectors  $\{\bx_i:i\in A\}$ and $\{\bx_i:i\in B\}$ are pairwise distinct. Otherwise, the objective function  should be minimized only by  distinct auxiliaries in $\{\bx_i:i\in A\}\cup\{\bx_i:i\in B\}$, and $n$ is the corresponding size of the pooled set.   The values for the two tuning parameters $\lambda_1$ and $\lambda_2$ are determined by  five-fold cross validation.

Intuition for the objective function (\ref{eq: object2}) is briefly discussed. First, $S(\bgamma,u)$ in (\ref{eq:S2}) balances  the two estimators for the population mean $N^{-1}\sum_{i=1}^Nu(\bx_i)$ over $u\in\mathcal{H}$. As discussed in Section~\ref{sec: basic setup}, since the sampling weights are incorporated for the probability sample $B$, the estimator  $N^{-1}\sum_{i=1}^N\delta_{B,i} \pi_{B,i}^{-1}u(\bx_i)$  is design-unbiased for $N^{-1}\sum_{i=1}^Nu(\bx_i)$ and $u\in\mathcal{H}$ \citep{horvitz1952generalization}. On the other hand, if $1+(Nn_A^{-1}-1)r_i$ is close to $\omega_i^{\star}$ for $i\in A$, the first term in (\ref{eq:S2}) is also approximately design-unbiased, so $S(\bgamma,u)$ should be small. However, $S(\bgamma,u)$ is not scale invariant, and $S(\bgamma,cu) = c^2S(\bgamma,u)$ for $c\in\mathbb{R}$. Thus, to make it scale-invariant, we consider $S(\bgamma,u)/\lVert u\rVert^2_2$ in the objective function (\ref{eq: object2}). Since $\mathcal{H}$ is large, there may not exist $\bgamma$ such that $S(\bgamma,u)=0$ holds for every $u\in\mathcal{H}$, so we uniformly balance the two estimators by  minimizing $\sup_{u\in\mathcal{H}}\{S(\bgamma,u)/\lVert u\rVert^2_2\}$. Because we have assumed that  $m(\bx)$ is a smooth function in (\ref{eq: popu model}),  a penalty on the smoothness of a function $u$ is incorporated in the objective function (\ref{eq: object2}). To stabilize the estimated weights, we use $-\lambda_2{Q}_A(\bgamma)$ in (\ref{qq0}) as another penalty.

\begin{rem}
We highlight the difference between the proposed method and existing ones. \citet{chen2020doubly} proposed a parametric model for $\pr(\delta_{A,i}=1\mid \bx_i)$, and the model parameters are estimated by a  calibration method based on a pre-specified smooth function $h(\bx)$. There exist different choices for $h(\bx)$, including basis functions of P-splines \citep{breidt2005model}, neural network estimators \citep{montanari2005}, and other modern machine learning methods \citep{breidt2017model}.  Even though the aforementioned works were not proposed to adjust the selection bias for a non-probablity sample, their methods can be easily implemented in the framework of \citet{chen2020doubly}. 
Rather than calibrating the predetermined basis function as existing works, we proposed to uniformly calibrate functions in an RKHS, and the limiting properties of the proposed estimator are also investigated; see Section~\ref{sec: asymptot} for details.   
\end{rem}

\begin{rem}\label{rem: KL diver}
\citet{wong2018kernel} used $\lambda_2N^{-1}\sum_{i=1}^N\delta_{A,i}\{1+(Nn_A^{-1}-1)r_i\}^2$ as a penalty to avoid extremely large sampling weights, and it is similar to the second term in (\ref{eq: I II PRIME}). For such a penalty, as $\lambda_2\to\infty$, all estimated sampling weights are close to 1. Then,  $\bar{Y}_N$ is estimated merely by the mean $n_A^{-1}\sum_{i\in A}y_i$ of the non-probability sample $A$, and this estimator is biased since the selection mechanism  for the non-probability $A$   is overlooked. To avoid this possible estimation bias when $\lambda_2$ is large, we propose $-\lambda_2{Q}_A(\bgamma)$ instead. Then, as $\lambda_2\to\infty$, we can still get reasonable estimates for the sampling weight  due to the fact that the density ratio function $r^\star( \bx)$ is the maximizer of $Q (r)$ in (\ref{eq: popu level}), which can be unbiased estimated by ${Q}_A(\bgamma)$ in (\ref{qq0}). Numerical results demonstrate the superior performance of the new objective function compared with \citet{wong2018kernel}; see  Section~\ref{sec: simu} for details.
\end{rem}

By the representer theorem, the solution of the inner optimization of (\ref{eq: object2}) lies in the space spanned by $\{K(\bx_i,\cdot):i\in A\cup B\}$, where $K(\bx,\by)$ is the kernel function associated with the RKHS $\mathcal{H}$. We can adopt a similar procedure as in Section~2.3 of \citet{wong2018kernel} to solve the optimization problem (\ref{eq: object2}); see Section~\ref{append: deri} of the Supplementary Material for details.
Once $\{\hat{r}_i:i\in A\}$ are obtained, we get $ \hat{\omega}_i = 1+ ( N n_A^{-1}- 1 )\hat{r}_i
 $ for $i\in A$, and the parameter of interest $\bar{Y}_N$ can be  estimated by 
\begin{equation}
\hat{Y}_N = N^{-1}\sum_{i\in A}\hat{\omega}_iy_i.\label{eq: proposed estimator kl}
\end{equation}
Asymptotic properties of the proposed estimator $\hat{Y}_N$ in (\ref{eq: proposed estimator kl}) are discussed in the next section. 

{
\begin{rem}
\citet{pmlr-v80-hebert-johnson18a} and \citet{kim2022universal} proposed a multicalibration framework to estimate $m(\bx)$ in (\ref{eq: popu model}), and they showed that their method can be applied to analyze different target (sub-)populations. Even though their methods are not discussed under non-probability sampling, they essentially use uniform calibration. Different from \citet{pmlr-v80-hebert-johnson18a} and \citet{kim2022universal}, our proposed method can be regarded as multitask-oriented. Although we explicitly propose a regression model in (\ref{eq: popu model}), the response of interest is not involved in the objective function (\ref{eq: object2}). Thus, a single set of  the estimated sampling weights $\{\hat{\omega}_i:i\in A\}$ can be  applied to different $Y$ variables and the internal consistency among survey estimates can be achieved  in the non-probability sample. 
\end{rem}
}

\section{Asymptotic theory}\label{sec: asymptot}
Since   $\mathcal{X}$ is compact, we set $\mathcal{X}=[0,1]^d$ for simplicity. A Sobolev space is commonly used when the underlying regression function is smooth, and by Proposition~12.31 of \citet{wainwright2019high},  we  consider a tensor product RKHS $\mathcal{H}=\bigotimes_{j=1}^d\mathcal{H}_j$, where $\mathcal{H}_j$ is an $l$-th order Sobolev space, $$W^{l,2}[0,1] = \{f:f,f^{(1)},\ldots,f^{(l-1)}\mbox{ are absolutely continuous}, f^{(l)}\in L^2([0,1])\},$$
    and $f^{(k)}$ is the $k$-th derivative of a function $f$ for $k=1,\ldots,l$. The corresponding reproducing kernel of $\mathcal{H}$ is $K(\bz_1,\bz_2) = \prod_{j=1}^dK_s(z_{1j},z_{2j})$, where $\bz_i= (z_{i1},\ldots,z_{id})^{\T}\in[0,1]^d$ for $i=1,2$, and $K_s(\cdot,\cdot)$ is the  reproducing kernel of $W^{l,2}[0,1]$ \citep[Section~1.2]{wahba1990spline}. See Section~12.2 of \citet{wainwright2019high} for discussion about other reproducing kernels. For $u\in\mathcal{H}$, we also assume  $\lVert u\rVert_{\mathcal{H}}^2<\infty$  for the sequential analysis.
    
To investigate the theoretical properties of the proposed method, we adopt the asymptotic framework of \citet{isaki1982survey} and consider a sequence of finite populations and samples. Besides, we make the following additional assumptions.
\begin{enumerate}
 \renewcommand{\labelenumi}{A\arabic{enumi}.}
 \setcounter{enumi}{3}
     \item The true regression function $m\in\mathcal{H}$ and $d/l< 2$.\label{ass:ratio}
    \item There exist positive constants $C_{\sigma,1}<C_{\sigma,2}$, $\delta$ and $C_{\delta}$ with respect to $N$, such that $C_{\sigma,1}\leq \sigma_i^2\leq C_{\sigma,2}$ and $E\{\lvert\epsilon_i\rvert^{2+\delta}\}<C_{\delta}$ for $i=1,\ldots,N$. Besides, the errors terms $\{\epsilon_i:i=1,\ldots, N\}$ are independent with the sampling indicators $\{\delta_{B,i}:i=1,\ldots, N\}$ for the probability sample $B$. \label{ass: epsilon}
    \item  The  rejective sampling design satisfies $N^{-1}\sum_{i=1}^N(\delta_{B,i}\pi_{B,i}^{-1}-1)y_i = O_p(n_B^{-1/2})$.\label{ass: sample B convergence result}
    \item There exist positive constants $C_{B,1}\leq C_{B,2} $ with respect to $n_B$ and $N$, such that $C_{B,1}\leq \pi_{B,i}Nn_B^{-1}\leq C_{B,2}$ for $i=1,\ldots,N$. Besides,   $n_BN^{-1}=o(1)$ and $N^{1/2}n_B^{-1}=o(1)$. \label{ass: prob sample inc prob}
    \item There exists a positive constant $M^\star$ such that $\lVert r^\star\rVert_\mathcal{H}\leq M^\star$.\label{ass: r bound}
\end{enumerate}

% \begin{assumption}\label{ass: sample size B}%
% There exists $\alpha\in(0,1)$, such that $n_B\asymp N^{\alpha}$.
% \end{assumption}

In Assumption~A\ref{ass:ratio}, we assume that the true regression function $m$ lies in $\mathcal{H}$, so it can be well approximated by a certain function in $\mathcal{H}$. Besides, the assumption $d/l<2$  regulates the complexity of the RKHS to guarantee theoretical properties of the proposed method;  see  Lemma~S6 of \citet{wong2018kernel} for details. The first part of Assumption~A\ref{ass: epsilon} is a common condition to show the limiting distribution of the proposed estimator. Since the responses of interest $\{y_1,\ldots,y_N\}$ are not available in the probability sample $B$, it is reasonable to postulate independence between the response of interest and the sampling indicators for the probability sample $B$ in the second part of Assumption~A\ref{ass: epsilon}. Assumption~A\ref{ass: sample B convergence result} guarantees the convergence rate of the estimator $N^{-1}\sum_{i=1}^N\delta_{B,i}\pi_{B,i}y_i$, and it is also a common assumption in survey sampling; see Section~1.3.2 of \citet{fuller2009} for details. 
Assumption~A\ref{ass: prob sample inc prob} is widely used to regulate $\pi_{B,i}$; see Theorem~1.3.5 of \citet{fuller2009} and condition C5 of \citet{chen2020doubly} for details. Rather than assuming that $n_B$ has asymptotically the same order as $N$ as in \citet{breslow2007}, \citet{han2021} and other references on empirical process for survey sampling, we make a more practically reasonable assumption  $n_B=o(N)$ for a probability sample  in Assumption~A\ref{ass: prob sample inc prob}. The technical condition $N^{1/2}n_B^{-1}=o(1)$ guarantees the convergence rate in Lemma~\ref{lemma: R S2} below; see Lemma~\ref{lemma: B part exponential inequality} in Section~\ref{supp: proof of lemma R S2} of the Supplementary Material for details. In Assumption~A\ref{ass: prob sample inc prob}, we implicitly assume that $n_B$ is non-stochastic, and we should use $o_p(1)$ instead if such an assumption fails. Even though we can show that $r^\star(\bx)$ is bounded by Assumption~A\ref{ass:piB}, Assumption~A\ref{ass: r bound} is a condition on its smoothness, and a similar condition is also implicitly assumed in (10) of \citet{nguyen2010}.

\begin{lem}\label{lemma: R S2}
Suppose that Assumptions~A\ref{ass: A ind}--A\ref{ass:ratio} and Assumption~A\ref{ass: prob sample inc prob} hold.  Then, there exists a positive constant $c$, such that for all $T\geq c$, 
\begin{equation*}
    P\left\{\sup_{u\in\widetilde{\mathcal{H}}_N}\frac{n_BS_N(\bgamma^{\star},u)}{\lVert u\rVert^{d/l}_{\mathcal{H}}}\ge T^2\right\} \leq c\exp\left(-\frac{16T^2}{c^2}\right),
\end{equation*}
where $\bgamma^{\star} = (r_1^\star,\ldots,r_N^\star)^{\T}$ and $\tilde{\mathcal{H}}_N =\{u\in\mathcal{H}:\lVert u\rVert_2=1\}$.
\end{lem}
The proof of Lemma~\ref{lemma: R S2} is relegated to Section~\ref{supp: proof of lemma R S2} of the Supplementary Material. Lemma~\ref{lemma: R S2} is a counterpart of Lemma~S1 of \citet{wong2018kernel}, and it establishes the convergence rate of $S_N(\bgamma^{\star},u)$ when the true density ratios $\{r_i^\star:i=1,\ldots,N\}$ are available. In addition, it serves as a building block to investigate the consistency and the limiting distribution of the proposed estimator. Rather than assuming the availability of $\{\bx_i:i=1,\ldots,N\}$ in \citet{wong2018kernel}, we consider the case when only a probability sample $\{\bx_i:i\in B\}$ is available. Besides, under rejective sampling, the sampling indicators $\{\delta_{B,i}:i=1,\ldots,N\}$ are negatively associated, so we develop a different proof to incorporate the design features from the probability sample $B$.

\begin{thm}\label{theorem: convergence rate of proposed estimator}
Suppose that Assumptions~A\ref{ass: A ind}--A\ref{ass: r bound}, $\lambda_1\asymp n_B^{-1}$ and $\lambda_2\asymp n_B^{-1}$  hold. Then, we have 
\begin{equation*}
    N^{-1}\sum_{i=1}^N(\delta_{A,i}\hat{w}_{i}-1)y_i = O_p(n_B^{-1/2}).
\end{equation*}
\end{thm}

Theorem~\ref{theorem: convergence rate of proposed estimator} establishes the consistency of the estimator in (\ref{eq: proposed estimator kl}), and its proof is in Section~\ref{supp: proof theorem convergence rate} of the Supplementary Material. By Theorem~\ref{theorem: convergence rate of proposed estimator}, $\hat{Y}_N$ in (\ref{eq: proposed estimator kl}) achieves  a parametric convergence rate $O_p(n_B^{-1/2})$, even though we do not assume  any parametric models for $m(\bx)$ in (\ref{eq: popu model}) and the selection mechanism for the non-probability sample $A$. 
{The proposed estimator in  (\ref{eq: proposed estimator kl}) is consistent by Theorem~\ref{theorem: convergence rate of proposed estimator}, but it is hard to derive an unbiased variance estimator for it. Instead, we propose to use the bootstrap variance estimator discussed in \cite{kim2019hypothesis}; see Section~\ref{supp: bootstrap variance estimator} of the Supplementary Material for details.  } 

\begin{thm}\label{theorem: CLT}
Suppose that Assumptions~A\ref{ass: A ind}--A\ref{ass: r bound} hold. Let $h=\hat{m} - m\in\mathcal{H}$ such that $\lVert h\rVert_\mathcal{H}=O_p(1)$, $\lVert h\rVert_2=o_p(1)$, $\lambda_2\lVert h\rVert_2^2=o_p(n_B^{-1})$, $\lambda_1=o(n_B^{-1})$ and $\lambda_1^{-1}\lVert h\rVert_2^{2(2l-d)/d}=o_p(n_B)$, where $\hat{m}$ is a kernel estimator of $m$. Then, we have 
\begin{equation*}
   B_N^{-1}\left\{\sum_{i=1}^N(\delta_{A,i}\hat{w}_{i}-\delta_{B,i}\pi_{B,i}^{-1})y_i - \sum_{i=1}^N(\delta_{A,i}\hat{w}_i -\delta_{B,i}\pi_{B,i}^{-1})\hat{m}(\bx_i)\right\} \to N(0,1)
\end{equation*}
in distribution, where $B_N^2 = \sum_{i=1}^N(\delta_{A,i}\hat{w}_i-\delta_{B,i}\pi_{B,i}^{-1})^2\sigma_i^2$. In addition, $B_N\asymp N^2 n_B^{-1/2}$  in probability.
\end{thm}

Theorem~\ref{theorem: CLT} establishes the limiting distribution for the proposed method, and its proof is relegated to Section~\ref{supp: proof theorem CLT} of the Supplementary Material.  We have validated the convergence rate of $\hat{Y}_N$ in Theorem~\ref{theorem: convergence rate of proposed estimator}, but it is hard to get its limiting distribution. Instead, we propose the following ``calibrated'' estimator, 
\begin{equation}
    \hat{Y}_{prop} = N^{-1}\sum_{i=1}^N\delta_{B,i}\pi_{B,i}^{-1}\hat{m}(\bx_i) + N^{-1}\sum_{i=1}^N\delta_{A,i}\hat{w}_i\{y_i - \hat{m}(\bx_i)\}. \label{eq: estimator}
\end{equation}
By Theorem~\ref{theorem: CLT} and the fact that $N^{-1}\sum_{i=1}^N\delta_{B,i}\pi_{B,i}^{-1}y_i$ is design-unbiased for $\bar{Y}_N$, we conclude that  $\hat{Y}_{prop}$ is asymptotically unbiased. The proposed estimator $\hat{Y}_{prop}$ is similar to a doubly robust estimator \citep{chen2020doubly}, but it is more attractive since that we do not make any model assumption for the regression model $m(\bx)$ in (\ref{eq: popu model}) and the response model $\pi_A(\bx)$.  
The corresponding variance estimator of $\hat{Y}_{prop}$ is established in the following corollary.

\begin{cor}\label{cor: variance estimator}
Suppose that the assumptions in Theorem~\ref{theorem: CLT} hold. Then, a plug-in variance estimator of $ \hat{Y}_{prop}$ in (\ref{eq: estimator}) is 
\begin{equation*}
    \hat{V}_{prop} = \hat{V}\left\{N^{-1}\sum_{i=1}^N\delta_{B,i}\pi_{B,i}^{-1}\hat{m}(\bx_i) \right\} + N^{-2}\sum_{i=1}^N\delta_{A,i}\hat{w}_i^2\{y_i-\hat{m}(\bx_i)\}^2,
\end{equation*}
where $\hat{V}\left\{N^{-1}\sum_{i=1}^N\delta_{B,i}\pi_{B,i}^{-1}z_i \right\}$ is a design-based variance estimator of $N^{-1}\sum_{i=1}^N\delta_{B,i}\pi_{B,i}^{-1}z_i$ for a fixed sequence $\{z_1,\ldots,z_N\}$. 
\end{cor}
The proof of Corollary~\ref{cor: variance estimator} is relegated to Section~\ref{ss: cor1} of the Supplementary Material. The first term of $\hat{V}_{prop}$ estimates the variability due to probability sampling, and the second term estimates $N^{-2}B_N^2$ in Theorem~\ref{theorem: CLT}.

\section{Simulation study}\label{sec: simu}
In this section, the performance of the proposed estimator (\ref{eq: proposed estimator kl}) is compared with its alternatives in terms of estimating the population mean $\bar{Y}_N$. The finite population $\{(y_i,\bx_i):i=1,\ldots,N\}$ and the two samples $A$ and $B$ are generated by the following setups. 
\begin{enumerate}
    \item[Linear.] $m(\bx_i)=10+2x_{1i}+2x_{2i}$ and $\epsilon_i\sim N(0,1)$, where $\bx_i=(x_{1i},x_{2i})^{\T}$, $x_{1i}=z_{1i}$, $x_{2i}=0.3x_{1i}+z_{2i}$, $z_{ki} = 2(\xi_{ki}-0.5)$ for $k=1,2$, $\xi_{ki}\sim\mbox{Beta}(3,3)$, and $\mbox{Beta}(\alpha,\beta)$ is a beta distribution with two shape parameters $\alpha$ and $\beta$. A non-probability sample $A$ is generated by Assumptions~\ref{ass: A ind}--\ref{ass: MAR}, where  $\pi_A(\bx_i)\propto \{m(\bx_i) - m_{min} + 0.25\}$, $\sum_{i=1}^N\pi_A(\bx_i) = n_{A0}$, $m_{min} = \min\{m(\bx_i):i=1,\ldots,N\}$, and $n_{A0}$ is the expected size of the non-probability sample $A$.

    \item[Nonlinear.] $m(\bx_i) = 3 + 2z_{1i} + z_{2i}$ and  $\epsilon_i\sim N(0,0.5^2)$,  where $\bx_i=(x_{1i},x_{2i})^{\T}$, $x_{1i} =\lvert z_{1i}\rvert\exp(-z_{1i})$, $x_{2i} = \lvert z_{2i}\rvert\exp(z_{2i})$, and $z_{1i}$ and $z_{2i}$  are independently generated by a truncated normal distribution restricted on the interval $[-3,3]$ with mean zero and standard deviation one. The response probability of the non-probability sample $A$ is $\logit\{c_A\pi_{A}(\bx_i;\btheta_0)\} =  1 -0.8z_{1i} -0.8 z_{2i}$, where $c_A$ is chosen such that $\sum_{i=1}^N\pi_{A}(\bx_i;\btheta_0)=n_{A0}$. 
\end{enumerate}
For each setup, we conduct Poisson sampling to generate a probability sample $B$, and the corresponding including probability satisfies $\pi_{B,i} \propto \log\{m(\bx_i)-m_{min} + 2\}$ and $\sum_{i=1}^N\pi_{B,i} =  n_{B0}$,  where  $n_{B0}$ is the expected size of the probability sample $B$.
The linear model setup is similar to \citet{chen2020doubly}, but the selection mechanism for the non-probability sample $A$ is not based on a logistic regression model.  A nonlinear regression model is considered in the second setup, and it is similar to \citet{wong2018kernel}. Even though we adopt a logistic regression model for the selection mechanism of the non-probability sample $A$, it is not linear in  $\bx_i$. 

We consider $(N, n_{A0},n_{B0})\in\{ (5\,000, 1\,000,100), (10\,000, 2\,000,200)\}$, and the following estimators are compared:
\begin{enumerate}
    \item Naive sample mean (NSM): $\hat{Y}_{NSM}=n_A^{-1}\sum_{i\in A}y_i$.
    \item Quasi-randomization estimator \citep{elliott2017inference} with $N$ known (EV1): $\hat{Y}_{EV1}=N^{-1}\sum_{i\in A}\tilde{w}_iy_i$, where $\tilde{w}_i=\tilde{d}_i\tilde{p}_i$, $\tilde{d}_i$ is obtained by a linear regression model for $d_{B,i}$ against $\bx_i$ based on the probability sample $B$, $\tilde{p}_i = \hat{P}(\delta_{A,i}\mid \bx_i, \delta_{A,i}+\delta_{B,i}\geq  1)/\hat{P}(\delta_{B,i}\mid \bx_i, \delta_{A,i}+\delta_{B,i}\geq1)$, and $\hat{P}(\delta_{A,i}\mid \bx_i, \delta_{A,i}+\delta_{B,i}\geq  1)$ and $\hat{P}(\delta_{B,i}\mid \bx_i, \delta_{A,i}+\delta_{B,i}\geq1)$ are estimated by a logistic regression model; see \citet{elliott2017inference} and \citet[Section~11.2]{kim2013statistical} for details.
    \item Quasi-randomization estimator \citep{elliott2017inference} with $N$ estimated (EV2): $\hat{Y}_{EV2}=(\sum_{i\in A}\tilde{w}_i)^{-1}\sum_{i\in A}\tilde{w}_iy_i$, where $\tilde{w}_i$ is estimated in the same way as EV1.
    \item Doubly robust estimator \citep{chen2020doubly} with $N$ known (DR1); see Section~\ref{ss: DRE} of the Supplementary Material for details.
     \item Doubly robust estimator \citep{chen2020doubly} with $N$ estimated (DR2).
   \item HT estimator (\ref{eq: proposed estimator kl}) with the proposed penalty (HT\_KL). 
   \item Balancing estimator of \citet{wong2018kernel} adapted to survey sampling (BSS). That is, instead of using the KL-divergence as the penalty term, we consider
     \begin{equation}
	\hat{\bgamma} = \argmin_{\xi_1\leq r_i\leq  \min\{\xi_2,C_{N}\}}\left[\sup_{u\in{\mathcal{H}}}\left\{\frac{S(\bgamma,u)}{\lVert u\rVert_2^2}-\lambda_1\frac{\lVert u\rVert_{\mathcal{H}}^2}{\lVert u\rVert_2^2}\right\}+ \lambda_2 Q_2 (\bgamma) \right],\label{eq: object222}
\end{equation}
where $Q_2(\bgamma) = n_A^{-1}\sum_{i\in A} \{1 + (Nn_A^{-1}-1)r_i\}^2$. 
     \item Proposed  estimator in (\ref{eq: estimator}) (Prop).
\end{enumerate}
For the two doubly robust estimators, we make an assumption as in \citet{chen2020doubly} that the underlying regression model $m(\bx)$ is linear in the auxiliary vector $\bx$ and the response model  $\pi_{A}(\bx)$ is logistic. It is worth pointing out that the BSS estimator has not been proposed by other researchers yet, and we use an $l_2$ penalty for the sampling weights in the proposed method for comparison.

We conduct $M=1\,000$ Monte Carlo simulations for each model setup, and Figure~\ref{fig: 1} shows the  Monte Carlo bias of the corresponding estimators, where the Monte Carlo bias  is 
$
    \hat{Y}_N^{(m)}-\bar{Y}_N^{(m)}
$
for $m=1,\ldots,M$, $\hat{Y}_N^{(m)}$ is a specific estimator with respect to  the $m$-th Monte Carlo simulation, and $\bar{Y}_N^{(m)}$ is the corresponding finite population mean. Regardless of model setups, NSM is biased since it fails to incorporate the selection mechanism for the non-probability sample. When the regression model is correctly specified,  EV2 and DR2 with population size estimated are more efficient than the others.  HT\_KL  is slightly more efficient than EV1, DR1, BSS and Prop, but it has a positive bias, especially when the size of the non-probability sample is large.  Prop is nearly as efficient as EV1, DR1 and BSS. However, when the regression model $m(\bx)$ and the response model $\pi_{A}(\bx)$ are wrongly specified,  all estimators other than BSS and Prop are biased, regardless of the sample sizes. A similar phenomenon for the doubly robust estimators was also discussed by \citet{kang2007demystifying}. Besides, the efficiency gain by EV2 and DR2 is less compared with their counterparts. {Although we have established the consistency of   HT\_KL  in Theorem~\ref{theorem: convergence rate of proposed estimator}, its finite sample performance is questionable when the true model is complex.} On the contrary, both BSS and Prop are unbiased. Compared with BSS, Prop is slightly more efficient, especially for the nonlinear model setup. 

\begin{figure}[!th]
    \centering
    \includegraphics[width = \textwidth]{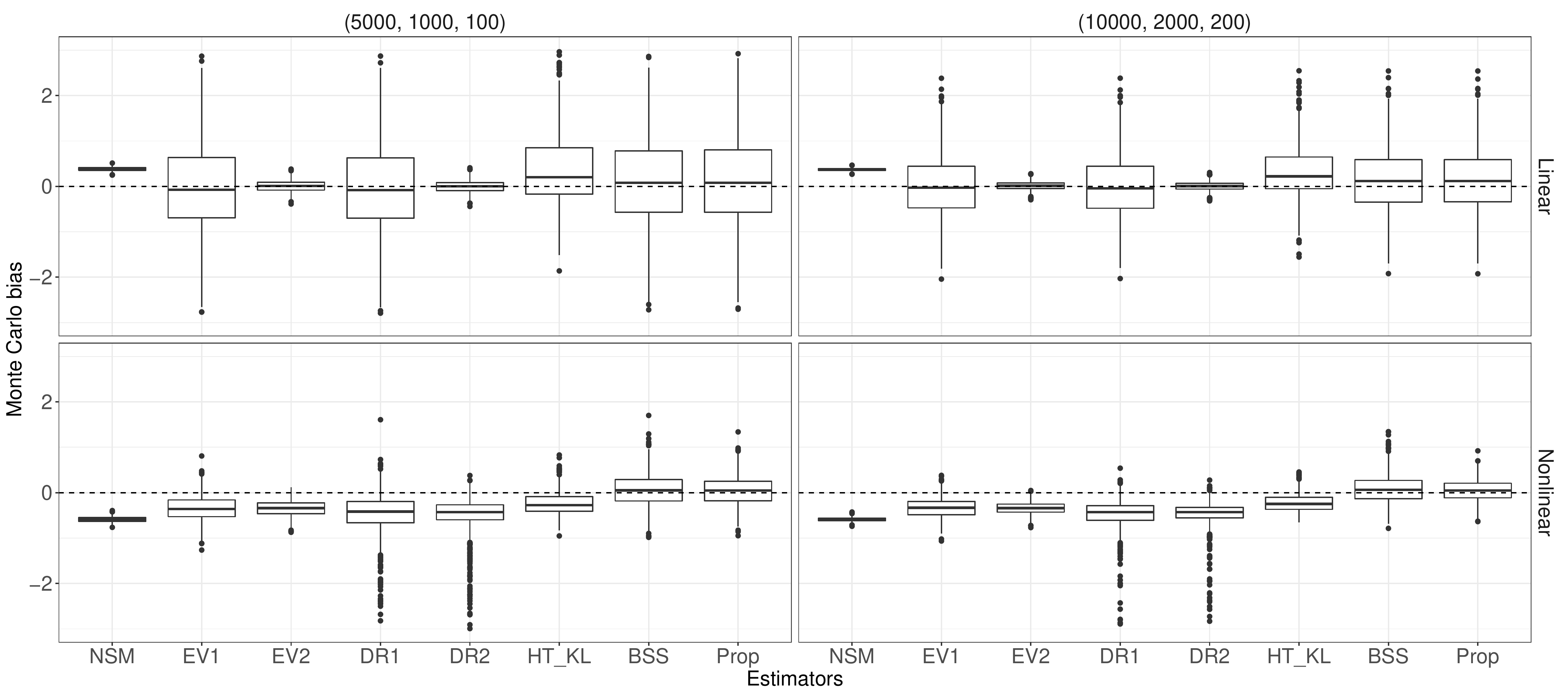}
    \caption{Boxplots for the Monte Carlo bias of different estimators under different setups. The horizontal dashed line corresponds to no bias.}
    \label{fig: 1}
\end{figure}

We also compare the computation efficiency of BSS and Prop, and the average computation time required by each estimator is shown in Table~\ref{tab: 2}. Regardless of the model setups, the proposed estimator (\ref{eq: estimator}) is more computationally efficient than BSS.

%The following results are based on Simu_v3_parallel.R, Auxiliary_v3.R and Analysis_v3.R. These files can be found in the code folder of the server.
\begin{table}
\centering
\caption{Computation efficiency of BSS and Prop in terms of average computation time based on 1\,000 Monte Carlo simulations (unit: second). }\label{tab: 2}
\begin{center}
\begin{tabular}{crrrrr}
  \hline
 \multirow{2}{*}{Sizes}& \multicolumn{2}{c}{Linear} & &\multicolumn{2}{c}{Nonlinear}\\
 &BSS & Prop &&BSS & Prop\\ 
  \hline
 (5\,000, 1\,000,100) &  35.84 &   7.26 &&31.11 &   5.22  \\ 
 (10\,000, 2\,000,200) & 66.59 &  31.03&& 137.69 &  42.47  \\ 
   \hline
\end{tabular}
\end{center}
\end{table}

The coverage rates of the interval estimator with 95\% confidence level is also investigated for Prop, and we use Corollary~\ref{cor: variance estimator} to estimate its variance. Specifically, under Poisson sampling, a plug-in variance estimator is 
\begin{equation*}
    \hat{V}_{prop} = N^{-1}\sum_{i=1}^N\delta_{B,i}\pi_{B,i}^{-2}(1-\pi_{B,i})\hat{m}(\bx_i)^2  +\hat{\sigma}^2 N^{-2}\sum_{i=1}^N\delta_{A,i}\hat{w}_i^2,
\end{equation*}
where  $\hat{m}(\bx)$ is obtained by the generalized additive model \citep{Wood2003}, and $\hat{\sigma}^2$ is the sample variance of $\{y_i-\hat{m}(\bx_i):i\in A\}$. The coverage rates are close to 0.95 in different model setups, especially when the sample sizes are large.  
%The following results are based on Simu_v3_parallel.R, Auxiliary_v3.R and Analysis_v3.R. These files can be found in the code Dropbox folder.
\begin{table}
    \centering
    \caption{Coverage rate of the interval estimator with 95\% confidence level based on 1\,000 Monte Carlo simulations. }
    \label{tab:my_label}
    \begin{center}
    \begin{tabular}{ccc}
         \hline 
       Model  &(5\,000, 1\,000,100) & (10\,000, 2\,000,200)\\
         \hline
         Linear&0.960&0.959\\
         Nonlinear&0.967&0.966\\
         \hline
    \end{tabular}
    \end{center}
\end{table}

{
\begin{rem}
Although the performance of   HT\_KL  is questionable under the nonlinear model setup, we still consider the performance of its bootstrap variance estimator, and the number of bootstrap replication is $B=200$. We relegate the simulation results to Section~\ref{ss: bootstrap variance estimator} of the Supplementary Material, and its performance is satisfactory in terms of relative bias.
\end{rem}
}

\section{Application} \label{sec: application}
We compare the performance of the proposed estimator and its alternatives based on a non-probability sample $A$ from the National Health Insurance Sharing Service (NHISS) and a probability sample $B$ from the Korea National Health and Nutrition Examination Survey (KNHANES). National Health Insurance was implemented in 1963 by the Medical Insurance Act, and whole Korean citizens are virtually enrolled in building a healthcare system; see  \citet{choi2015effect}, \citet{jee2019gender} and the references within for details. KNHANES, on the other hand, is a national survey conducted by the  Korea Centers for Disease Control and Prevention since 1998, and it is mainly adopted to assess the health and nutrition status of Korean citizens and provide health-related statistics in Korea. Therefore, the sample of KNHANES is nationally representative, and health-related information, including socioeconomic status, quality of life, health-related behaviors, and healthcare utilization, has been collected; see \citet{kweon2014data} and the references within for details. 

In this section, a non-probability sample $A$ contains $n_A=20\,000$ elements randomly selected from an NHISS dataset.  Demographic information, including age and gender, and health-related information, such as total cholesterol (mg/dL), hemoglobin (HGB), triglyceride (TG), and high-density lipoprotein cholesterol (HDL, mg/dL), is available. The probability sample $B$ is a subset of the blood test results in the 2014 KNHANES, and it was obtained by a multi-stage clustered probability design with sample size  $n_B = 4\,929$. The probability sample $B$ contains the health-related information as that in the non-probability sample. We are interested in estimating the average total cholesterol for different age and gender groups by incorporating information from the two samples.  

Even though the average total cholesterol can be estimated by $\hat{N}^{-1}\sum_{i\in B}\pi_{B,i}^{-1}y_i$ with $\hat{N}=\sum_{i\in B}\pi_{B,i}^{-1}$, we treat it as unavailable and use it as a benchmark to evaluate the performance of different methods, where $y_i$ is the total cholesterol for the $i$th person. That is, we only assume $\{(\bx_i,y_i):i\in A\}$ and $\{(\bx_i,\pi_{B,i}):i\in B\}$ are available, where $\bx_i$ contains the covariates, including HGB, TG and HDL. The population size $N$ is not available, and we use $\hat{N}$ instead. Both samples can be categorized into three age groups, including 20--40, 40--60, and more than 60 years old. Table~\ref{tab: marginal app} summarizes the marginal means of the covariates within each age and gender group, and we conclude that there exists a difference for the covariates in the two samples. 

\begin{table}[ht]
\centering
\caption{Marginal means of covariates for the non-probability sample $A$ and the probability sample $B$ for different domains. ``20--40'' stands for the group with age between 20 and 40, ``40--60'' for the group with age between 40 and 60, and ``60+'' for the group with age more than 60.}\label{tab: marginal app}
\begin{center}
\begin{tabular}{ccrrrrrr}
  \hline
\multirow{2}{*}{Gender} & \multirow{2}{*}{Covariate} & \multicolumn{2}{c}{20--40}&\multicolumn{2}{c}{40--60} & \multicolumn{2}{c}{60+ } \\
&& $A$ & $B$ & $A$ & $B$ & $A$ & $B$\\
  \hline 
  \multirow{3}{*}{Female} & HDL & 63.82 & 57.28 &  59.74 &  54.72 &  55.20 &  49.90 \\ 
   & TG & 83.72 & 89.16 & 110.68 & 118.31 & 128.28 & 137.86 \\ 
   & HGB & 12.95 & 13.06 &  12.96 &  13.24 &  12.89 &  13.18 \\ 
   &  &  &  &  &  &  &  \\ 
\multirow{3}{*}{Male} & HDL &  53.14 &  48.19 &  51.78 &  47.12 &  51.09 &  46.95 \\ 
   & TG & 147.21 & 160.74 & 162.87 & 184.98 & 133.50 & 141.65 \\ 
   & HGB &  15.47 &  15.68 &  15.18 &  15.41 &  14.39 &  14.58 \\ 
   \hline
\end{tabular}
\end{center}
\end{table}

For each age and gender group, consider a regression model (\ref{eq: popu model}) for the proposed method, and the corresponding benchmark is
$
\hat{Y}_{BM} = (\sum_{i\in D}\pi_{B,i}^{-1})^{-1}\sum_{i\in D}\pi_{B,i}^{-1}y_i,
$
where $D\subset B$ consists of elements in the group. We also consider NSM, EV2 and DR2 in Section~\ref{sec: simu} for comparison, and different methods are evaluated by the  estimation error $\hat{Y} -\hat{Y}_{BM}$, where $\hat{Y}$ is  a specific estimator. 

Figure~\ref{fig: estimation error} summarizes the estimation errors of different methods for each age and gender group, and we can reach the following conclusions.  NSM overestimates the average total cholesterol for each age and gender group, and its performance is questionable. Even though  HT\_KL performs better than NSM, it is still much worse than EV2, DR2, BSS and Prop. Prop  and BBS perform at least as well as EV2 and DR2 in all groups, and they outperform EV2 and DR2 for some groups. For example, in the female group with age 20--40, the estimation errors of Prop and BSS are less than EV2 and DR2, and a similar observation holds for the male group with age 20--40. 

\begin{figure}[!ht]
    \centering
    \includegraphics[width=\textwidth]{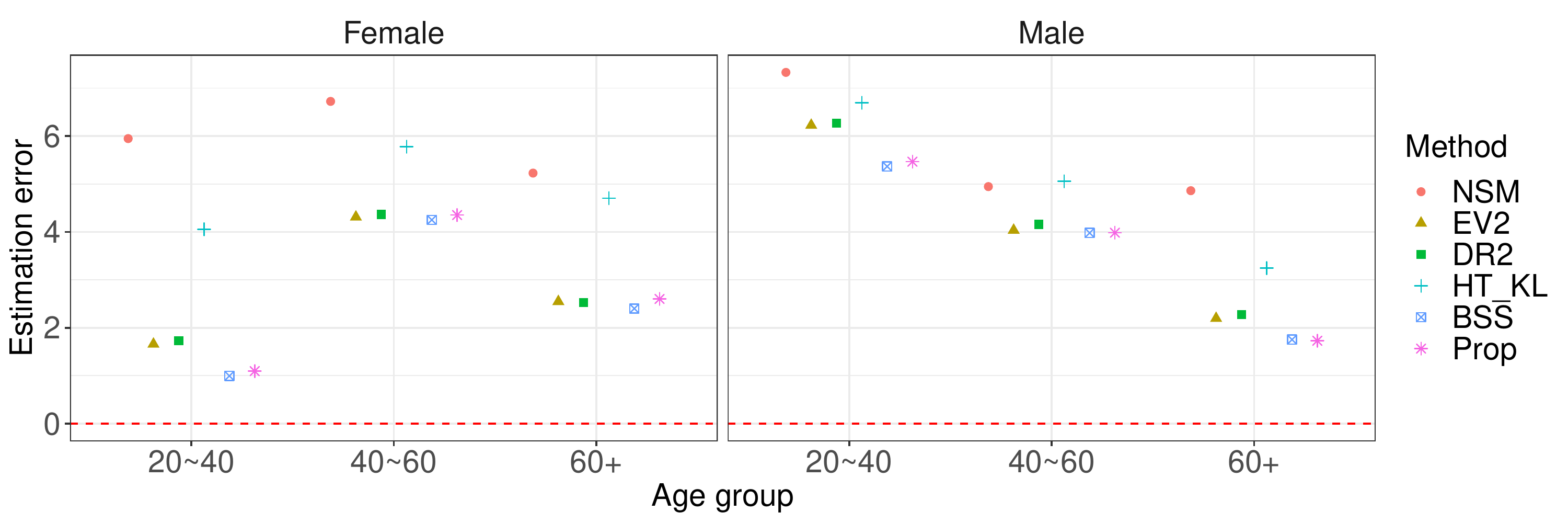}
    \caption{Estimation error of different methods for each age and gender group. ``NSM'' is the naive sample mean estimator using the non-probability sample, ``EV2'' is the method considered by \citet{elliott2017inference}, ``DR2'' is the one proposed by \citet{chen2020doubly},``HT\_KL'' is the estimator in (\ref{eq: proposed estimator kl}), ``BSS'' is the balancing estimator of \citet{wong2018kernel} adopted to survey sampling, and ``Prop'' is the proposed method.}% The following result is based on code_bootstrap_conf_int_v5.R and Analysis_conf_int_v5.R
    \label{fig: estimation error}
\end{figure}

\section{Concluding remarks}\label{sec: conclusion}
We propose a uniform function calibration method to estimate the sampling weights of a non-probability sample based on a probability sample, which is generated by a rejective sampling design. Compared with existing methods, the proposed method does not make any parametric assumption either for the regression model or the response model, so it can be widely adopted in practice. Besides, different from existing works, a KL-divergence-based penalty is proposed to improve the performance of the proposed method. 
Consistency and the asymptotic normality of the proposed estimator are established under regularity conditions. Numerical results show that the proposed method outperforms its alternatives, especially when both regression and response models are wrongly specified. 
The proposed method can be viewed as a ``soft'' calibration method, since we do not require that $S(\hat{\bgamma},u)=0$ holds for every $u\in\mathcal{H}$. In survey sampling, however, we may would like to achieve ``hard'' calibration for certain functions of the covariates, it would be an interesting project to incorporate a ``hard'' calibration component in the proposed objective function. 

\bibliographystyle{apalike}

\bibliography{reference}  
\newpage
\begin{center}
	\LARGE\bfseries Supplemental Material for ``Functional Calibration under Non-Probability Survey Sampling''
\end{center}

\renewcommand{\theequation}{S\arabic{section}.\arabic{equation}}
\renewcommand{\thelem}{S\arabic{lem}}
\renewcommand{\thethm}{S\arabic{thm}}
\setcounter{equation}{0}
\setcounter{thm}{0}
\setcounter{lem}{0}
\setcounter{section}{0}
\renewcommand{\thesection}{S\arabic{section}}

\section{Brief introduction to RKHS}\label{ss RKH}
A symmetric bivariate function $K:\mathcal{X}\times\mathcal{X}\to \mathbb{R}$ is a positive semidefinitive kernel function, if for all integer $n\geq1$ and elements $\{\bx_1,\ldots,\bx_n\}\subset \mathcal{X}$, the $n\times n$ matrix $\bM$ is positive semidefinitive, where  $K(\bx_i,\bx_j)$ serves as its $(i,j)$-th entry for $i,j=1,\ldots,n$; see Definition~12.6 of \citet{wainwright2019high} for details.

Let $K(\bx,\by)$ be a positive semidefinitive kernel function, and consider a functional space $${\mathcal{H}}^\dagger = \left\{f:f(\cdot)=\sum_{i=1}^n\alpha_jK(\cdot,\bx_j)\mbox{ for some }n\geq1, \{\alpha_1,\ldots,\alpha_n\}\subset\mathbb{R}, \{\bx_1,\ldots,\bx_n\}\subset\mathcal{X}\right\}.$$Then, by Theorem~12.11 of \citet{wainwright2019high}, the complement of ${\mathcal{H}}^\dagger$, say $\mathcal{H}$, is an RKHS with the reproducing kernel $K(\bx,\by)$.

Furthermore, suppose that the kernel function $K(\bx,\by)$ has the following eigen-decomposition:
$$
K(\bx,\by) = \sum_{j=1}^\infty \mu_j\psi_j(\bx)\psi_j(\by),
$$
where $\{\mu_j:j=1,2,\ldots\}$ are non-negative eigenvalues satisfying $\sum_{j=1}^\infty\mu_j^2<\infty$, and $\{\psi_j(\bx):j=1,2,\ldots\}$ are the corresponding eigenfunctions. Then, for any $f\in\mathcal{H}$, there exist $\{c_j:j=1,2,\ldots\}$ such that 
$$
f(\bx) = \sum_{j=1}^\infty c_j\psi_j(\bx),
$$
and the corresponding norm associated with $\mathcal{H}$ is defined as 
$$
\lVert f\rVert_\mathcal{H}^2 = \sum_{j=1}^\infty c_j^2/\mu_j.
$$
See Section~12.2.3 of \citet{wainwright2019high} and Section~5.8.1 of \citet{friedman2001elements} for details.

\section{Numerical solution of the optimization problem}\label{append: deri}
Consider $u(\bx) = \sum_{i=1}^n\alpha_iK(\bx_i,\cdot)$, and denote $\bu = \bM\balpha$, where $\bu = (u(\bx_1),\ldots,u(\bx_n))^{\T}$, $\balpha=(\alpha_1,\ldots,\alpha_n)^{\T}$, and  $\bM$ is an  $n\times n$ Gram matrix with $(i,j)$-th element being $K(\bx_i,\bx_j)$. Assume that the eigen-decomposition of $\bM$ is 
\begin{equation}
	\bM = \begin{pmatrix}\bP_1&\bP_2\end{pmatrix}\begin{pmatrix}\bQ_1&\bzero\\\bzero&\bQ_2\end{pmatrix}\begin{pmatrix}\bP_1^{\T}\\\bP_2^{\T}\end{pmatrix},
\end{equation}
where $\bQ_1$ is a diagonal matrix consisting the positive eigenvalues, and $\bQ_2$ is a zero matrix. Notice that  $\bQ_2$ may be a null matrix. Then, 	$S(\bgamma,u) = N^{-2}\balpha^{\T}\bM\bA(\bgamma)\bM\balpha$ and $\lVert u\rVert_2^2 = n^{-1}\balpha^{\T}\bM^2\balpha$, where $\bA(\bgamma) = \bw(\bgamma)\bw(\bgamma)^{\T}$, and $\bw(\bgamma) = (w_1(\bgamma),\ldots,w_n(\bgamma))^{\T}$. Denote $\bbeta = n^{-1/2}\bQ_1\bP_1^{\T}\balpha$, the inner optimization problem becomes 
\begin{equation}
	\sup_{\bbeta:\lVert\bbeta\rVert_2\leq 1}\bbeta^{\T}\left\{\frac{n}{N^2}P_1^{\T}A(\bgamma)P_1 - n\lambda_1Q_1^{-1}\right\}\bbeta.
\end{equation}
Then, we can use a similar procedure in Section~2.3 of \citet{wong2018kernel} to solve the optimization problem (\ref{eq: object2}).

\section{Proof of Lemma~\ref{lemma: R S2}} \label{supp: proof of lemma R S2}
% Let $u^{\star}=\arg\min_{u\in\widetilde{\mathcal{H}}_{N}}\{S_{N}(\br^{\star},u)-\lambda_1\norm{u}_{\mathcal{H}}^2\}$ Denote $\mathcal{H}_{1}=\{u\in\mathcal{H}:\lVert u\rVert_\mathcal{H}\leq 1\}$.
First, we present a definition of negative association as in Definition~2.1 of \citet{joag1983negative}.
\begin{defi}\label{def: NA}
	Random variables $X_1,\ldots,X_N$ are said to be negatively associated if for every pair of disjoint subsets $A_1$, $A_2$ of $\{1,\ldots,N\}$, 
	\begin{equation}\label{eq: negative associated inequality}
		\mathrm{Cov}\{f_1(X_i,i\in A_1),f_2(X_i,i\in A_2)\}\leq 0,
	\end{equation}
	where $f_1$ and $f_2$ are increasing functions in every variable.
\end{defi}

It can be easily shown that if both $f_1$ and $f_2$ are decreasing functions in every variable, we still get (\ref{eq: negative associated inequality}) for negatively associated random variables since $-f_1$ and $-f_2$ are increasing and $\text{Cov}\{-f_1(X_i,i\in A_1),-f_2(X_i,i\in A_2)\}=\text{Cov}\{f_1(X_i,i\in A_1),f_2(X_i,i\in A_2)\}$.

\begin{lem}\label{lemma: exp W2}
	Suppose that Assumption~A\ref{ass: prob sample inc prob} holds. Then, we have 
	$$
	E \{\exp(W_{B,i}^2)\}-1\leq \sigma_0^2,
	$$
	uniformly for $i=1,\ldots,N$, where $W_{B,i} = (\delta_{B,i}\pi_{B,i}^{-1}-1)n_BN^{-1}$ and $\sigma_0^2 = \exp\{\max\{1,(C_{B,2}^{-1}-1)^{2},C_{B,1}^{-2}\}\}-1$. 
\end{lem}
\begin{proof}[Proof of Lemma~\ref{lemma: exp W2}]
	Consider 
	\begin{eqnarray}
		E \{\exp(W_{B,i}^2)\} 
		&=& \pi_{B,i}\exp\{(\pi_{B,i}^{-1}-1)^2n_B^2N^{-2}\}+(1-\pi_{B,i})\exp(n_B^2N^{-2})\notag \\ 
		&\leq&\pi_{B,i}\exp(\max\{(C_{B,2}^{-1}-1)^{2},C_{B,1}^{-2}\}) + (1-\pi_{1})e\notag \\ 
		&\leq&\exp(\max\{1,(C_{B,2}^{-1}-1)^{2},C_{B,1}^{-2}\}),\label{eq: W2 1}
	\end{eqnarray}
	where first inequality holds by Assumption~A\ref{ass: prob sample inc prob}.
	By (\ref{eq: W2 1}), we have proved Lemma~\ref{lemma: exp W2}.
\end{proof}

The next lemma is a straightforward result from Definition~\ref{def: NA}, so we omit its proof. 
\begin{lem} \label{lemma: NA decreasing}
	For any $m\geq 2$ and mutually disjoint subsets $A_1,\ldots,A_m$ of $\{1,\ldots,N\}$ and negatively associated random variables $X_1,\ldots,X_N$,  we have 
	\begin{equation}
		E\left\{\prod_{k=1}^mf_k(X_{i}:i\in A_k)\right\}\leq \prod_{k=1}^mE\left\{f_k(X_{i}:i\in A_k)\right\},  
	\end{equation}
	where $f_1,\ldots,f_m$ are decreasing and non-negative functions in every variable.
\end{lem}
The next lemma shows a Hoeffding's inequality \citep[Lemma 3.5]{geer2000empirical} under rejective sampling. 
\begin{lem}\label{lemma: Hoeffding}
	Suppose Assumption~A\ref{ass: prob sample inc prob} holds. Then, there exist positive constants $C_{B,3}$ and $C_{B,4}$, such that for any $a\geq 0$ and  $\{\gamma_i:i=1,\ldots,N\}\subset\mathbb{R}$, we have 
	\begin{equation}
		P\left(\left\lvert\sum_{i=1}^NW_{B,i}\gamma_i\right\rvert\geq a\right)\leq C_{B,3}\exp\left\{- \frac{ a^2}{C_{B,4}\sum_{i=1}^N\gamma_i^2}\right\},\label{eq: P final}
	\end{equation}
	where $W_{B,i}$ is defined in Lemma~\ref{lemma: exp W2}.
\end{lem}
\begin{proof}[Proof of Lemma~\ref{lemma: Hoeffding}]
	Since the probability sample $B$ is generated by a rejective sampling design, the corresponding sampling indicators are negatively associated by Theorem~3 of \citet{bertail2016sharp}. Given the sequence  $\{\gamma_i:i=1,\ldots,N\}$, denote $\mathcal{I} = \{i:\gamma_i\geq 0\}$. Then, for any positive constant $a$, we have 
	\begin{eqnarray}
		P\left(\left\lvert\sum_{i=1}^NW_{B,i}\gamma_i\right\rvert\geq a\right) &\leq& P\left(\left\lvert\sum_{i\in\mathcal{I}}W_{B,i}\gamma_i\right\rvert\geq \frac{a}{2}\right) +  P\left(\left\lvert\sum_{i\notin\mathcal{I}}W_{B,i}\gamma_i\right\rvert\geq \frac{a}{2}\right)\notag \\
		&\leq& P\left(\sum_{i\in\mathcal{I}}W_{B,i}\gamma_i\geq \frac{a}{2}\right) + 
		P\left(\sum_{i\in\mathcal{I}}W_{B,i}\gamma_i\leq -\frac{a}{2}\right)\notag \\ 
		&&+P\left(\sum_{i\notin\mathcal{I}}W_{B,i}\gamma_i\geq \frac{a}{2}\right)
		+P\left(\sum_{i\notin\mathcal{I}}W_{B,i}\gamma_i\leq -\frac{a}{2}\right).\notag\\\label{eq: prob 12}
	\end{eqnarray}
	Assume $ P\left(\left\lvert\sum_{i\in\mathcal{I}}W_{B,i}\gamma_i\right\rvert\geq a/2\right)=0$ if $\mathcal{I}=\emptyset$ and $ P\left(\left\lvert\sum_{i\notin\mathcal{I}}W_{B,i}\gamma_i\right\rvert\geq a/2\right)=0$ if $\mathcal{I}=\{1,\ldots,N\}$. 
	
	Without loss of generality, assume that $\mathcal{I}\neq\emptyset$ and $\mathcal{I}\neq\{1,\ldots,N\}$. For any $\beta>0$, we have
	\begin{eqnarray}
		P\left(\sum_{i\in\mathcal{I}}W_{B,i}\gamma_i\geq \frac{a}{2}\right)
		&\leq&\exp\left(- \frac{\beta a}{2}\right)E\left\{\exp\left(\beta\sum_{i\in\mathcal{I}}W_{B,i}\gamma_i\right)\right\}\notag \\ 
		&\leq&\exp\left(- \frac{\beta a}{2}\right)\prod_{i\in\mathcal{I}}E\left\{\exp\left(\beta W_{B,i}\gamma_i\right)\right\}\notag \\ 
		&\leq&\exp\left\{2(1+\sigma_0^2)\beta^2\sum_{i\in\mathcal{I}}\gamma_i^2- \frac{\beta a}{2}\right\}\notag \\ 
		&\leq&\exp\left\{2(1+\sigma_0^2)\beta^2\sum_{i=1}^N\gamma_i^2- \frac{\beta a}{2}\right\},\label{eq: first P}
	\end{eqnarray}
	where the first inequality holds by Cramer's inequality \citep[Corollary~3.1.5]{athreya2006measure}, the second inequality holds by Property~P2 of \citet{joag1983negative} and the fact that $\beta\gamma_i>0$ for $i\in \mathcal{I}$, the third inequality by Lemma~8.1 of \citet{geer2000empirical} and Lemma~\ref{lemma: exp W2}, the last inequality holds since $2(1+\sigma_0^2)>0$, and  $\sigma_0^2$ is defined in Lemma~\ref{lemma: exp W2}. If we set 
	$$
	\beta = \frac{a}{8(1+\sigma_0^2)\sum_{i=1}^N\gamma_i^2},
	$$
	by (\ref{eq: first P}), we have 
	\begin{equation}
		P\left(\sum_{i\in\mathcal{I}}W_{B,i}\gamma_i\geq \frac{a}{2}\right)\leq \exp\left\{- \frac{ a^2}{32(1+\sigma_0^2)\sum_{i=1}^N\gamma_i^2}\right\}.\label{eq: P part one final}
	\end{equation}
	
	Next, for any $\beta>0$, consider 
	\begin{eqnarray}
		P\left(\sum_{i\in\mathcal{I}}W_{B,i}\gamma_i\leq -\frac{a}{2}\right)&=&P\left(\sum_{i\in\mathcal{I}}W_{B,i}\tilde{\gamma}_i\geq \frac{a}{2}\right)\notag\\
		&\leq&  \exp\left(- \frac{\beta a}{2}\right)E\left\{\exp\left(\beta\sum_{i\in\mathcal{I}}W_{B,i}\tilde{\gamma}_i\right)\right\}\notag \\ 
		&\leq&\exp\left(- \frac{\beta a}{2}\right)\prod_{i\in\mathcal{I}}E\left\{\exp\left(\beta W_{B,i}\tilde{\gamma}_i\right)\right\}\notag
	\end{eqnarray}
	where $\tilde{\gamma}_i=-\gamma_i$, and the last inequality holds by Lemma~\ref{lemma: NA decreasing} since $\beta\tilde{\gamma}_i\leq 0$ for $i\in\mathcal{I}$.  Then, we can use a similar argument leading to (\ref{eq: P part one final}) to get
	\begin{equation}
		P\left(\sum_{i\in\mathcal{I}}W_{B,i}\gamma_i\leq -\frac{a}{2}\right)\leq \exp\left\{- \frac{ a^2}{32(1+\sigma_0^2)\sum_{i=1}^N\gamma_i^2}\right\}.\label{eq: P part one final 2}
	\end{equation}
	Besides, we can also get 
	\begin{eqnarray}
		P\left(\sum_{i\notin\mathcal{I}}W_{B,i}\gamma_i\geq \frac{a}{2}\right)\leq \exp\left\{- \frac{ a^2}{32(1+\sigma_0^2)\sum_{i=1}^N\gamma_i^2}\right\},\label{eq: P part one final 3}
	\end{eqnarray}
	since $\gamma_i\leq0$ for $i\notin\mathcal{I}$.
	\begin{equation}
		P\left(\sum_{i\notin\mathcal{I}}W_{B,i}\gamma_i\leq -\frac{a}{2}\right) = P\left(\sum_{i\notin\mathcal{I}}W_{B,i}\tilde{\gamma}_i\geq \frac{a}{2}\right).\notag
	\end{equation}
	Then, we can use a similar procedure as (\ref{eq: first P})--(\ref{eq: P part one final}) to verify 
	\begin{equation}
		P\left(\sum_{i\notin\mathcal{I}}W_{B,i}\gamma_i\leq -\frac{a}{2}\right)\leq \exp\left\{- \frac{ a^2}{32(1+\sigma_0^2)\sum_{i=1}^N\gamma_i^2}\right\},\label{eq: P part one final 3}
	\end{equation}
	since $\tilde{\gamma}_i\geq0$ for $i\notin\mathcal{I}$.
	
	By (\ref{eq: prob 12}) and (\ref{eq: P part one final})--(\ref{eq: P part one final 3}), we have proved Lemma~\ref{lemma: Hoeffding} with $C_{B,3}=4$ and $C_{B,4}=32(1+\sigma_0^2)$.
\end{proof}

Let $\mathcal{H}_1=\{u\in\mathcal{H}:\lVert u\rVert_\mathcal{H}=1\}$. By Lemma~S7 of \citet{wong2018kernel}, there exists a constant $R$ such that 
\begin{equation}
	\sup_{u\in\mathcal{H}_1}\lVert u\rVert_\infty\leq R.\label{eq: h1 uniformbound}
\end{equation} Denote $H_\infty(\epsilon,\mathcal{H}_1)$ to be the uniform entropy for $\mathcal{H}_1$; see Definition~2.3 of \citet{geer2000empirical} for details about the uniform entropy.  

\begin{lem} \label{lemma: Entropy inequality}
	Suppose Assumption~A\ref{ass:ratio} holds. There exists a constant $C_{B,5}$, such that for any $n_B\leq N$ and  $S\geq S_0$, we have 
	\begin{eqnarray}
		\sum_{s=S_0}^S2^{-s}RH_\infty^{1/2}(2^{-s}(n_BN^{-1})^{1/2}R,\mathcal{H}_1)\leq C_{B,5}N^{1/2}n_B^{-1/2},
	\end{eqnarray}
	where $S_0=\max\{s:R\leq 2^{-s}n_BN^{-1}\leq 2R\}$. 
\end{lem}
\begin{proof}[Proof of Lemma~\ref{lemma: Entropy inequality}]
	By Assumption~A\ref{ass:ratio} and Lemma~S6 of \citet{wong2018kernel}, there exists a constant $C_{\mathcal{H}}$ such that for $\epsilon>0$,
	\begin{equation}\label{eq: entropy bound infty}
		H_\infty(\epsilon,\mathcal{H}_1)\leq C_{\mathcal{H}}\epsilon^{-d/l}.
	\end{equation}
	
	Consider 
	\begin{eqnarray}
		&&\sum_{s=S_0}^S2^{-s}RH_\infty^{1/2}(2^{-s}(n_BN^{-1})^{1/2}R,\mathcal{H}_1)\notag \\
		&=&(Nn_B^{-1})^{1/2}\sum_{s=S_0}^S2^{-s}(n_BN^{-1})^{1/2}RH_\infty^{1/2}(2^{-s}(n_BN^{-1})^{1/2}R,\mathcal{H}_1)\notag\\
		&\leq&2(Nn_B^{-1})^{1/2}\int_0^{2R}H_\infty^{1/2}(\epsilon,\mathcal{H}_1)\mbox{d}\epsilon\notag \\ 
		&\leq&2(Nn_B^{-1})^{1/2}\frac{C_{\mathcal{H}}^{1/2}(2R)^{1-d/(2l)}}{1-d/(2l)}\notag \\ 
		&=&C_{B,5}(Nn_B^{-1})^{1/2},\label{eq: entropy result lemma 4}
	\end{eqnarray}
	where the second inequality holds by (\ref{eq: entropy bound infty}) and $C_{B,5} = 2^{2-d/(2l)}C_{\mathcal{H}}^{1/2}R^{1-d/(2l)}\{1-d/(2l)\}^{-1}$. Thus, we have proved Lemma~\ref{lemma: Entropy inequality} by (\ref{eq: entropy result lemma 4}).
\end{proof}

\begin{lem}\label{lemma: B part exponential inequality}
	Suppose Assumption~A\ref{ass:ratio} and Assumption~A\ref{ass: prob sample inc prob} hold. Then,  for all $0\leq \epsilon<\delta$ and $K>1$, there exists $N_0=N_0(\delta,\epsilon)$, such that for $N\geq N_0$, we have
	\begin{equation*}
		P\left[\left\{\sup_{u\in\mathcal{H}_1}\left\lvert\frac{1}{n_B}\sum_{i=1}^NW_{B,i}u(\bx_i)\right\rvert\geq \delta\right\}\bigcap\left\{\left\lvert\frac{1}{n_B}\sum_{i=1}^NW_{B,i}\right\rvert\leq K\right\}\right]\leq C_{B,6}\exp\left\{-\frac{n_B(\delta-\epsilon)^2}{C_{B,6}R^2}\right\},
	\end{equation*}
	where  $C_{B,6}$ only depends  on $\sigma_0^2$ defined in Lemma~\ref{lemma: exp W2}.
\end{lem}
\begin{proof}[Proof of Lemma~\ref{lemma: B part exponential inequality}]

	By (\ref{eq: entropy bound infty}), there exists a finite $N_s$ such that
	$\{u_j^s:j=1,\ldots,N_s\}$ is  a minimal $\{2^{-s}(n_BN^{-1})^{1/2}R\}$-covering set of $\mathcal{H}_1$ for $s=S_0,S_0+1,\ldots,S$ in terms of the $\lVert\cdot\rVert_\infty$ norm, where $S_0=\max\{s:R\leq 2^{-s}n_BN^{-1}\leq 2R\}$ and $S=\min\{s\geq1:2^{-s}(n_BN^{-1})^{1/2}R\leq \epsilon/(2K)\}$.

	By Assumption~A\ref{ass: prob sample inc prob} and Lemma~\ref{lemma: Entropy inequality}, there exists $N_0=N_0(\delta,\epsilon)$, such that when $N\geq N_0$, 
	\begin{equation}
		n_B^{1/2}(\delta-\epsilon)\geq \left\{12C_{B,4}^{1/2}R\sum_{s=S_0+1}^S2^{-s}H_{\infty}^{1/2}(2^{-s}(n_BN^{-1})^{1/2}R,\mathcal{H}_1)\right\}\vee \left\{(1152\log2)^{1/2}C_{B,4}^{1/2}R\right\},\label{eq: upper bound}
	\end{equation}
	where $a\vee b=\max\{a,b\}$.
	
	We adopt the notation convenience from Section~3.2 of \citet{geer2000empirical} and index functions in $\mathcal{H}_1$ by  $\Theta$: $\mathcal{H}_1 = \{u_\theta:\theta\in\Theta\}$. Then, for any $u_\theta\in\mathcal{H}_1$, there exists $u_\theta^S$ such that $\lVert u_\theta-u_\theta^S\rVert_\infty\leq \epsilon/(3K)$. Thus, we have 
	\begin{eqnarray}
		&&\left\lvert \frac{1}{n_B}\sum_{i=1}^{N} W_{B,i}\left\{u_\theta(\bx_i)-u_\theta^S(\bx_i)\right\} \right\rvert\notag\\
		&\leq& \frac{1}{n_B}\left\lvert\sum_{i=1}^N\delta_{B,i}\pi^{-1}_{B,i}\frac{n_B}{N}\left\{u_\theta(\bx_i)-u_\theta^S(\bx_i)\right\}\right\rvert + \frac{1}{N}\left\lvert\sum_{i=1}^N\left\{u_\theta(\bx_i)-u_\theta^S(\bx_i)\right\}\right\rvert\notag \\ 
		&\leq& \left\{\frac{1}{n_B}\left(\sum_{i=1}^N\delta_{B,i}\pi^{-1}_{B,i}\frac{n_B}{N}\right)+1\right\}\max_{i=1,\ldots,N}\left\lvert\left\{u_\theta(\bx_i)-u_\theta^S(\bx_i)\right\}\right\rvert\notag \\ 
		&\leq&\frac{K+2}{3K}\epsilon\leq \epsilon
	\end{eqnarray} on the event $\{n_B^{-1}\lvert\sum_{i=1}^{N} W_{B,i}\rvert\leq K\}$, where the first inequality holds by the definition of $W_{B,i}$ in Lemma~\ref{lemma: exp W2}, the third inequality is due to the fact that the event $\{n_B^{-1}\lvert\sum_{i=1}^{N} W_{B,i}\rvert\leq K\}$ implies $\{n_B^{-1}\sum_{i=1}^N\delta_{B,i}\pi^{-1}_{B,i}n_BN^{-1}\leq K+1\}$, and the last inequality holds since $K> 1$.
	
	Thus, it is enough to show the exponential inequality for 
	\begin{equation}\label{eq: remaining exp inequality}
		P\left\{\sup_{\theta\in\Theta}\left\lvert\frac{1}{n_B}\sum_{i=1}^{N}W_{B,i}u_\theta^S(\bx_i)\right\rvert\geq \delta - \epsilon\right\}.
	\end{equation}
	Define $u_\theta^{S_0}=0$ for $u\in\mathcal{H}_1$, and we have $u_\theta^S = \sum_{s=S_0+1}^S(u_\theta^s-u_\theta^{s-1})$. For any sequence $\{\eta_s:s=S_0+1,\ldots,S\}$ satisfying $\sum_{s=S_0+1}^S\eta_s\leq 1$, we have 
	\begin{eqnarray}
		&&P\left[\sup_{\theta\in\Theta}\left\lvert\frac{1}{n_B}\sum_{i=1}^N\sum_{s=S_0+1}^SW_{B,i}\left\{u_\theta^{s}(\bx_i)-u_\theta^{s-1}(\bx_i)\right\} \right\rvert\geq \delta-\epsilon\right]\notag \\ 
		&\leq& \sum_{s=S_0+1}^SP\left[\sup_{\theta\in\Theta}\left\lvert\frac{1}{n_B}\sum_{i=1}^NW_{B,i}\left\{u_\theta^{s}(\bx_i)-u_\theta^{s-1}(\bx_i)\right\}\right\rvert\geq \eta_s(\delta-\epsilon)\right]\notag \\ 
		&\leq&\sum_{s=S_0+1}^SC_{B,3}\exp\left\{2H_{\infty}(2^{-s}(n_BN^{-1})^{1/2}R,\mathcal{H}_1)-\frac{n_B\eta_s^2(\delta-\epsilon)^2}{9C_{B,4}2^{-2s}R^2}\right\},
	\end{eqnarray}
	where the last inequality holds by Lemma~\ref{lemma: Hoeffding} and 
	\begin{eqnarray}
		\lVert u_\theta^s-u_\theta^{s-1}\rVert_\infty
		&\leq& \lVert u_\theta^s-u_\theta\rVert_\infty + \lVert u_\theta^{s-1}-u_\theta\rVert_\infty\notag \\ 
		&\leq& 2^{-s}(n_BN^{-1})^{1/2}R +2^{-s+1}(n_BN^{-1})^{1/2}R \leq 3\{2^{-s}(n_BN^{-1})^{1/2}R\}.\notag
	\end{eqnarray}

	Now, we consider $$\eta_s = \frac{6RC_{B,4}^{1/2}2^{-s}H_{\infty}^{1/2}(2^{-s}(n_BN^{-1})^{1/2}R,\mathcal{H}_1)}{n_B^{1/2}(\delta-\epsilon)}\vee \frac{2^{-s+S_0}(s-S_0)^{1/2}}{8}.$$
	Then, by (\ref{eq: upper bound}) and a similar argument in the proof of Lemma~3.2 of \citet{geer2000empirical}, we can show that 
	$$
	\sum_{s=S_0+1}^S\eta_s \leq1.
	$$
	Thus, we can show that 
	\begin{eqnarray}
		&&\sum_{s=S_0+1}^SC_{B,3}\exp\left\{2H_{\infty}(2^{-s}(n_BN^{-1})^{1/2}R,\mathcal{H}_1)-\frac{n_B\eta_s^2(\delta-\epsilon)^2}{9C_{B,4}2^{-2s}R^2}\right\}\notag \\ 
		&&\leq \sum_{s=S_0+1}^SC_{B,3}\exp\left\{-\frac{n_B\eta_s^2(\delta-\epsilon)^2}{18C_{B,4}2^{-2s}R^2}\right\}.\notag
	\end{eqnarray}
	
	Thus, we can use a similar argument in the proof of Lemma~3.2 of \citet{geer2000empirical} to conclude the proof of Lemma~\ref{lemma: B part exponential inequality} by (\ref{eq: upper bound}) with $C_{B,6}=\max\{2C_{B,3},C_{B,4}\}$.
\end{proof}

By Lemma~\ref{lemma: Hoeffding} and setting $K$ sufficient large, we have 
\begin{equation}
	P\left(\left\lvert\frac{1}{n_B}\sum_{i=1}^NW_{B,i}\right\rvert> K\right)    \leq C_{B,6}\exp\left\{-\frac{n_B(\delta-\epsilon)^2}{C_{B,6}R^2}\right\},\label{eq: K part} 
\end{equation} 
where the related quantities are defined in Lemma~\ref{lemma: B part exponential inequality}.
Thus, by Lemma~\ref{lemma: B part exponential inequality} and (\ref{eq: K part}), we have 
\begin{eqnarray}
	&&P\left\{\sup_{u\in\mathcal{H}_1}\left\lvert\frac{1}{n_B}\sum_{i=1}^NW_{B,i}u(\bx_i)\right\rvert\geq \delta\right\}\notag \\ 
	&\leq&P\left[\left\{\sup_{u\in\mathcal{H}_1}\left\lvert\frac{1}{n_B}\sum_{i=1}^NW_{B,i}u(\bx_i)\right\rvert\geq \delta\right\}\bigcap\left\{\left\lvert\frac{1}{n_B}\sum_{i=1}^NW_{B,i}\right\rvert\leq K\right\}\right] + P\left(\left\lvert\frac{1}{n_B}\sum_{i=1}^NW_{B,i}\right\rvert> K\right) \notag \\ 
	&\leq& 2C_{B,6}\exp\left\{-\frac{n_B(\delta-\epsilon)^2}{C_{B,6}R^2}\right\}.
\end{eqnarray}
Thus, take $C_{B,7}=2C_{B,6}$ and $\epsilon$ sufficiently small, we conclude  
\begin{equation}\label{eq: unconditional}
	P\left\{\sum_{u\in\mathcal{H}_1}\left\lvert\frac{1}{n_B}\sum_{i=1}^NW_{B,i}u(\bx_i)\right\rvert\geq \delta\right\}\leq C_{B,7}\exp\left\{-\frac{n_B\delta^2}{C_{B,7}R^2}\right\}.
\end{equation}

Denote
\begin{eqnarray}
	S_{N,A}(\bgamma,u) &=& \left( N^{-1}\sum_{i=1}^{N} \left[\delta_{A,i}\left\{ 1+ \left( \frac{N}{n_A} -1\right)r_i \right\}-1\right]u(\bx_i) \right)^2,\label{eq:S2A}
\end{eqnarray}
and
\begin{eqnarray}
	S_{N,B}(u) &=& \left\{ N^{-1}\sum_{i=1}^{N} (\delta_{B,i} \pi^{-1}_{B,i}-1)u(\bx_i)\right\}^2.\label{eq:S2B}
\end{eqnarray}

\begin{proof}[Proof of Lemma~\ref{lemma: R S2}]
	By the fact that 
	$S_N(\bgamma^{\star},u)\leq 2S_{N,A}(\bgamma^{\star},u)+2S_{N,B}(u),\notag\label{eq: sup less than sup p sup}$
	we have 
	\begin{eqnarray}
		&&P\left\{\sup_{u\in\widetilde{\mathcal{H}}_N}\frac{n_BS_N(\bgamma^{\star},u)}{\lVert u\rVert^{d/l}_{\mathcal{H}}}\ge T^2\right\} \notag \\
		&\leq&P\left\{\sup_{u\in\widetilde{\mathcal{H}}_N}\frac{n_BS_{N,A}(\bgamma^{\star},u)}{\lVert u\rVert^{d/l}_{\mathcal{H}}}\ge \frac{T^2}{4}\right\}+
		P\left\{\sup_{u\in\widetilde{\mathcal{H}}_N}\frac{n_BS_{N,B}(u)}{\lVert u\rVert^{d/l}_{\mathcal{H}}}\ge \frac{T^2}{4}\right\},\label{eq: P less that P p P}
	\end{eqnarray}
	where $S_{N,A}(\bgamma,u)$ and $S_{N,B}(u)$ are in (\ref{eq:S2A}) and (\ref{eq:S2B}).
	By Assumptions~A\ref{ass: A ind}--A\ref{ass:ratio}, following a similar argument of Lemma~S1 of \citet{wong2018kernel}, we can show 
	\begin{equation}
		P\left\{\sup_{u\in\widetilde{\mathcal{H}}_N}\frac{NS_{N,A}(\bgamma^{\star},u)}{\lVert u\rVert^{d/l}_{\mathcal{H}}}\ge \frac{T^2}{4}\right\}\leq C_{A,3}\exp\left(-\frac{16T^2}{C_{A,3}^2}\right),\label{eq: part one P}
	\end{equation}
	where $C_{A,3}$ is a constant.  
	Since $n_B\leq N$, we have 
	\begin{eqnarray}
		P\left\{\sup_{u\in\widetilde{\mathcal{H}}_N}\frac{n_BS_{N,A}(\br^{\star},u)}{\lVert u\rVert^{d/l}_{\mathcal{H}}}\ge \frac{T^2}{4}\right\}&\leq&P\left\{\sup_{u\in\widetilde{\mathcal{H}}_N}\frac{NS_{N,A}(\br^{\star},u)}{\lVert u\rVert^{d/l}_{\mathcal{H}}}\ge \frac{T^2}{4}\right\}.\label{eq: part one P2}
	\end{eqnarray}
	By (\ref{eq: part one P})--(\ref{eq: part one P2}), we have 
	\begin{equation}
		P\left\{\sup_{u\in\widetilde{\mathcal{H}}_N}\frac{n_BS_{N,A}(\br^{\star},u)}{\lVert u\rVert^{d/l}_{\mathcal{H}}}\ge \frac{T^2}{4}\right\}\leq C_{A,3}\exp\left(-\frac{16T^2}{C_{A,3}^2}\right).\label{eq: part one Pfinal}
	\end{equation}

	Next, we investigate the second part on the right hand side of (\ref{eq: P less that P p P}). By  
	\begin{eqnarray}
		S_{N,B}(u) &=& \left[ \frac{1}{n_B}\sum_{i=1}^{N} W_{B,i}u(\bx_i)\right]^2,\label{eq:S2Bnon}
	\end{eqnarray}
	where $W_{B,i}$ is defined in Lemma~\ref{lemma: exp W2}.

	By (\ref{eq: unconditional}) and a similar proof for Lemma~8.4 of \citet{geer2000empirical}, there exists a constant $C_{B,8}$ such that 
	\begin{equation}
		P\left\{\sup_{u\in\widetilde{\mathcal{H}}_N}\frac{n_BS_{N,B}(u)}{\lVert u\rVert^{d/l}_{\mathcal{H}}}\ge \frac{T^2}{4}\right\}\leq C_{B,8}\exp\left(-\frac{16T^2}{C_{B,8}^2}\right).\label{eq: second part P}
	\end{equation}
	By (\ref{eq: part one Pfinal}) and (\ref{eq: second part P}), we have validated Lemma~\ref{lemma: R S2} by setting $c=C_{A,3}+C
	_{B,8}$.
\end{proof}

\section{Proof of Theorem~\ref{theorem: convergence rate of proposed estimator}}\label{supp: proof theorem convergence rate}
Recall $n=n_A+n_B$, and $\lVert u\rVert_2^2 = n^{-1}\sum_{i=1}^n u(\bx_i)^2$. Notice that we have assumed $\lVert u\rVert_{\mathcal{H}}<\infty$ for $u\in\mathcal{H}$.
\begin{lem}\label{eq: m bounded}
	Suppose Assumption~A\ref{ass:ratio} holds. Then, for $u\in\mathcal{H}$, we have 
	\begin{equation*}
		E(\lVert u\rVert_2^2)<\infty,\quad E(\lVert u\rVert_2^4)<\infty.
	\end{equation*}
\end{lem}
\begin{proof}[Proof of Lemma~\ref{eq: m bounded}]
	By Lemma~S7 of \citet{wong2018kernel}, there exists a constant $R$ such that $\sup_{u\in\mathcal{H}_1}\lVert u\rVert_\infty\leq R$. Thus, for $u\in\mathcal{H}$, $\lVert u\rVert_\infty \leq R\lVert u\rVert_{\mathcal{H}}$. Since we have assumed $\lVert u\rVert_\mathcal{H}<\infty$ for $u\in\mathcal{H}$ in Section~\ref{sec: asymptot}, we conclude that $\lVert u\rVert_\infty <\infty$, so we have proved Lemma~\ref{eq: m bounded}.
	
	%  Consider 
	%  \begin{eqnarray}
	%      \lVert u\rVert_2^2 = n^{-1}\sum_{i=1}^n u(\bx_i)^2
	%      \leq n_A^{-1}\sum_{i=1}^N \delta_{A,i}u(\bx_i)^2 + n_B^{-1}\sum_{i=1}^N\delta_{B,i}u(\bx_i)^2\leq 2\lVert u\rVert_\infty^2 .\label{eq: inequality 1}
	%  \end{eqnarray}
	% Thus, by (\ref{eq: inequality 1}), we can show $E(\lVert u\rVert_2^2)<\infty$. 
	
	% By basic inequality, we have 
	% \begin{equation}
	%     \lVert u\rVert_2^4 \leq 2\left\{n_A^{-1}\sum_{i=1}^N \delta_{A,i}u(\bx_i)^2\right\}^2 + 2\left\{n_B^{-1}\sum_{i=1}^N \delta_{B,i}u(\bx_i)^2\right\}^2 \leq 4 \lVert u\rVert_\infty^4.\label{eq: inequality 2}
	% \end{equation}
	% By (\ref{eq: inequality 2}), we have show $E(\lVert u\rVert_2^4)<\infty$.
\end{proof}

Denote $u^\star = \argmax_{u\in\tilde{\mathcal{H}}_N}\{S(\bgamma^\star,u)-\lambda_1\lVert u\rVert_\mathcal{H}\}$, and its existence is shown in Appendix~\ref{append: deri}, where $\widetilde{\mathcal{H}}_{N}=\{u\in\mathcal{H}:\norm{u}_{2}=1\}$.  Then, for any $u\in{\mathcal{H}}$, we have 
\begin{equation}\label{eq: f in H tilde space}
	S(\hat{\bgamma},u)-\lambda_1\lVert u\rVert_{\mathcal{H}}^2-\lambda_2Q_A(\hat{\bgamma})\lVert u\rVert_2^2\leq \{S(\bgamma^\star,u^\star)-\lambda_1\lVert u^\star\rVert_{\mathcal{H}}^2-\lambda_2Q_A(\bgamma^\star)\}\lVert u\rVert_2^2.
\end{equation}

\begin{lem}\label{lemma: S rhat m}
	Suppose Assumptions~A\ref{ass:piB}--A\ref{ass: prob sample inc prob} hold. If $\lambda_1\asymp n_B^{-1}$ and $\lambda_2\asymp n_B^{-1}$, we have $S(\hat{\bgamma},u) = O_p(n_B^{-1})\lVert u\rVert_2^2$ for $u\in\mathcal{H}$. 
\end{lem}
\begin{proof}[Proof of Lemma~\ref{lemma: S rhat m}]
	By (\ref{eq: f in H tilde space}), we have 
	\begin{equation} \label{eq: basic inequality S r m}
		S(\hat{\bgamma},u) +\lambda_1\lVert u^\star\rVert_{\mathcal{H}}^2\lVert u\rVert_2^2 + \lambda_2Q_A(\bgamma^\star)\lVert u\rVert_2^2 \leq 
		S(\bgamma^\star,u^\star)\lVert u\rVert_2^2 +\lambda_1\lVert u\rVert_{\mathcal{H}}^2 + \lambda_2Q_A(\hat{\bgamma})\lVert u\rVert_2^2.
	\end{equation}
	
	By Lemma~\ref{lemma: R S2}, we can use a similar argument in the proof of Lemma~S3 of \citet{wong2018kernel} to reach the following result:
	
	Case (i): Suppose that $S(\bgamma^\star,u^\star)\lVert u\rVert_2^2$ is the largest on the right-hand side of (\ref{eq: basic inequality S r m}). If $\lVert u\rVert_2^2>0$, we have $S(\hat{\bgamma},u)\leq \lambda_1^{-d/(2l-d)}O_p(n_B^{-2l/(2l-d)})\lVert u\rVert_2^2$. If $\lVert u\rVert_2^2=0$, we can still get the same result. 
	
	Case (ii): Suppose that $\lambda_1\lVert u\rVert_{\mathcal{H}}^2$ is the largest on the right-hand side of (\ref{eq: basic inequality S r m}). Then, we have $S(\hat{\bgamma},u)\leq 3\lambda_1\lVert u\rVert_{\mathcal{H}}^2$.
	
	Case (iii):  Suppose that $\lambda_2Q_A(\hat{\bgamma})\lVert u\rVert_2^2$ is the largest on the right-hand side of (\ref{eq: basic inequality S r m}). Then, we have $S(\hat{\bgamma},u)\leq 3\lambda_2Q_A(\hat{\bgamma})\lVert u\rVert_2^2$.
	
	By (\ref{eq: h1 uniformbound}) and the proof of Lemma~\ref{eq: m bounded}, $\lVert u\rVert_2 \leq \lVert u\rVert_\infty\leq R\lVert u\rVert_\mathcal{H}<\infty$. Then, we have 
	\begin{equation}
		S(\hat{\bgamma},u) = O_p\left[ \max\{  \lambda_1^{-d/(2l-d)}n_B^{-2l/(2l-d)}\lVert u\rVert_2^2,  \lambda_1\lVert u\rVert_{\mathcal{H}}^2,  \lambda_2Q_A(\hat{\bgamma})\lVert u\rVert_2^2\}   \right].\label{eq: rate of s gamma hat}
	\end{equation}
	By the proposed optimization problem (\ref{eq: object2}), $\hat{r}_i<\xi_2$ for $i=1,\ldots,N$, so $Q_A(\hat{\bgamma}) = O_p(1)$. Thus, by the conditions on $\lambda_1$ and $\lambda_2$, we have $S(\hat{\bgamma},u) = O_p(n_B^{-1})\lVert u\rVert_2^2$. Thus, we have complete the proof of Lemma~\ref{lemma: S rhat m}.
\end{proof}
Based on Lemma~\ref{eq: m bounded}, we can also use a similar argument as the proof for Lemma~S3 of \citet{wong2018kernel} to show that $E\{n_BS(\hat{\bgamma},u)\}<\infty$.   
\begin{proof}[Proof of Theorem~\ref{theorem: convergence rate of proposed estimator}]
	Consider 
	\begin{eqnarray}
		N^{-1}\sum_{i=1}^N(\delta_{A,i}\hat{w}_{i}-1)y_i &=& N^{-1}\sum_{i=1}^N(\delta_{A,i}\hat{w}_{i}-\delta_{B,i}\pi_{B,i}^{-1})y_i  + N^{-1}\sum_{i=1}^N(\delta_{B,i}\pi_{B,i}^{-1}-1)y_i\notag \\ 
		&=& N^{-1}\sum_{i=1}^N(\delta_{A,i}\hat{w}_{i}-\delta_{B,i}\pi_{B,i}^{-1})m_0(\bx_i)  +
		N^{-1}\sum_{i=1}^N(\delta_{A,i}\hat{w}_{i}-\delta_{B,i}\pi_{B,i}^{-1})\epsilon_i\notag \\ 
		&&+
		N^{-1}\sum_{i=1}^N(\delta_{B,i}\pi_{B,i}^{-1}-1)y_i.\label{eq: theorem 1 target}
	\end{eqnarray}
	By Assumption~\ref{ass:ratio} and Lemma~\ref{lemma: S rhat m}, we have 
	\begin{equation}\label{eq: Theorem1 part 1}
		N^{-1}\sum_{i=1}^N(\delta_{A,i}\hat{w}_{i}-\delta_{B,i}\pi_{B,i}^{-1})m_0(\bx_i)  = O_p(n_B^{-1/2}).
	\end{equation}
	Since $\{\epsilon_i:i=1,\ldots,N\}$ are independent with the sampling indicators $\{(\delta_{A,i},\delta_{B,i}):i=1,\ldots,N\}$ as well as the weights $\{(\hat{w}_i,\pi_{B,i}):i=1,\ldots,N\}$, by Assumption~A\ref{ass: epsilon}, we have 
	\begin{eqnarray}
		&&\var\left\{N^{-1}\sum_{i=1}^N\left(\delta_{A,i}\hat{w}_{i}-\delta_{B,i}\pi_{B,i}^{-1}\right)\epsilon_i\right\}\notag \\ 
		&\leq& C_{\sigma,2}E\left\{N^{-2}\sum_{i=1}^N\left(\delta_{A,i}\hat{w}_{i}-\delta_{B,i}\pi_{B,i}^{-1}\right)^2\right\}\notag \\ 
		&\leq&2C_{\sigma,2}E\left(N^{-2}\sum_{i=1}^N\delta_{B,i}\pi_{B,i}^{-2}\right)+2C_{\sigma,2}E\left(N^{-2}\sum_{i=1}^N\delta_{A,i}\hat{w}_i^2\right).\label{eq: Theorem 1 part 2}
	\end{eqnarray}
	
	To show the order of (\ref{eq: Theorem 1 part 2}), we first consider 
	\begin{eqnarray}
		E\left(N^{-2}\sum_{i=1}^N\delta_{B,i}\pi_{B,i}^{-2}\right)&=&N^{-2}\sum_{i=1}^N\pi_{B,i}^{-1}\notag \\ 
		&\leq&C_{B,1}^{-1}N^{-2}\sum_{i=1}^NNn_B^{-1}\notag \\ 
		&=& O(n_B^{-1}),\label{eq: Theorem1 part 21}
	\end{eqnarray}
	where the first inequality holds by Assumption~A\ref{ass: prob sample inc prob}.
	
	In addition, $n_A^{-1}(N-n_A)\to P(\delta=1)^{-1}P(\delta=0)$ almost surely by strong law of large number. Then, by Assumption~A\ref{ass:piB}, there exists a constant $C_{A,4}$, such that $n_A^{-1}(N-n_A)<C_{A,4}$ for $N\geq N_0$, where $N_0$ is determined by $C_{A,4}$. Then, for $N\geq N_0$, since $\hat{w}_i = 1+(N-n_A)/n_{A}\hat{r}_i$, we have
	\begin{eqnarray}
		N^{-2}\sum_{i=1}^N\delta_{A,i}\hat{w}_i^2
		&\leq&N^{-2}\sum_{i=1}^N\delta_{A,i}(1+C_{A,4}\hat{r}_i)^2 \notag \\ 
		&\leq&2N^{-2}\sum_{i=1}^N\delta_{A,i} + 2C_{A,4}^{2}N^{-2}\sum_{i=1}^N\delta_{A,i}\hat{r}_i^2\notag \\ 
		&\leq&2N^{-1} +2C_{A,4}^{2}N^{-2}\sum_{i=1}^N\xi_2^2\notag \\ 
		&=&2N^{-1}(1+C_{A,4}^2\xi_2^2).
		\label{eq: w2 subpart 2} 
		% &\leq&2N^{-1} + 2C^{2}N^{-2}\sum_{i=1}^N\delta_{A,i}I(\hat{r}_i\leq e^2)\hat{r}_i^2 + 2C^{2}N^{-2}\sum_{i=1}^N\delta_{A,i}I(\hat{r}_i>e^2)\hat{r}_i^2\notag \\
		% &\leq&2N^{-1}(1+C^2e^2)\notag \\ 
		% && + 2C^2N^{-2}\max\{\hat{r}_i\{\log(\hat{r}_i)-1\}^{-1}:\hat{r}_i>e^2\}\sum_{i=1}^N\delta_{A,i}I(\hat{r}_i>e^2)\hat{r}_i\{\log(\hat{r}_i)-1\}\notag \\ 
		% &\leq&2N^{-1}(1+C^2e^2) + 2C^{2}n_B^{-1}n_A^{-1}\sum_{i=1}^N\delta_{A,i}I(\hat{r}_i>N^{1/2}n_B^{-1/2})\hat{r}_i\{\log(\hat{r}_i)-1\}\notag\\\label{eq: w2 subpart 1} \\ 
		% &\leq&2N^{-1}(1+C^2e^2) + 2C^2n_B^{-1}Q_A(\hat{r}),\label{eq: w2 subpart 2} 
	\end{eqnarray}
	% where $\epsilon_0$ is defined in (\ref{eq: object2}), $I(x>C)=1$ if $x>C$ and 0 otherwise, (\ref{eq: w2 subpart 1}) holds since $\hat{r}_i\leq C_N$ in (\ref{eq: object2}) and  $f(x) = x\{\log(x)-1\}^{-1}$ is increasing for $x\geq e^2$. Thus, by Lemma~\ref{lemma: QA}, we have 
	Thus, by Assumption~A\ref{ass: prob sample inc prob} and (\ref{eq: w2 subpart 2}), we have 
	\begin{equation}
		E\left(N^{-2}\sum_{i=1}^N\delta_{A,i}\hat{w}_i^2\right) = O(n_B^{-1}).\label{eq: Theorem1 part 22}
	\end{equation}
	
	Thus, by (\ref{eq: Theorem 1 part 2})--(\ref{eq: Theorem1 part 21}) and (\ref{eq: Theorem1 part 22}), we have shown that 
	\begin{equation}\label{eq: theorem1 part2}
		\var\left\{N^{-1}\sum_{i=1}^N\left(\delta_{A,i}\hat{w}_{i}-\delta_{B,i}\pi_{B,i}^{-1}\right)\epsilon_i\right\} = O(n_B^{-1}).
	\end{equation}
	
	By (\ref{eq: theorem 1 target}), (\ref{eq: Theorem1 part 1}), (\ref{eq: theorem1 part2}) and Assumption~A\ref{ass: sample B convergence result}, we complete the proof of Theorem~\ref{theorem: convergence rate of proposed estimator}.
	
	% Next, consider 
	% \begin{eqnarray}
	%     &&\var\left\{N^{-1}\sum_{i=1}^N(\delta_{B,i}\pi_{B,i}^{-1}-1)y_i\right\}\notag \\ 
	%     &=&\var\left[E\left\{N^{-1}\sum_{i=1}^N(\delta_{B,i}\pi_{B,i}^{-1}-1)y_i\mid \mathcal{F}_N\right\}\right]+
	%     E\left[\var\left\{N^{-1}\sum_{i=1}^N(\delta_{B,i}\pi_{B,i}^{-1}-1)y_i\mid \mathcal{F}_N\right\}\right]\notag \\
	%     &=& E\left[\var\left\{N^{-1}\sum_{i=1}^N(\delta_{B,i}\pi_{B,i}^{-1}-1)y_i\mid \mathcal{F}_N\right\}\right],
	% \end{eqnarray}
	% where $\mathcal{F}_N=\{y_i:i=1,\ldots,N\}$, and the last equality hold since $E(\delta_{B,i}\pi_{B,i}^{-1}\mid \mathcal{F}_N)=1$.
	
	% Besides, we have 
	% \begin{eqnarray}
	%     &&\var\left\{N^{-1}\sum_{i=1}^N(\delta_{B,i}\pi_{B,i}^{-1}-1)y_i\mid \mathcal{F}_N\right\}\notag \\ 
	%     &=&N^{-2}\sum_{i=1}^N\pi_{B,i}^{-2}y_i^2\var(\delta_{B,i}\mid \mathcal{F}_N) + N^{-2}\sum_{i\neq j}\pi_{B,i}^{-1}\pi_{B,j}^{-1}y_iy_j\Cov(\delta_{B,i},\delta_{B,j}\mid \mathcal{F}_N)\notag \\ 
	%     &\leq&N^{-2}\sum_{i=1}^N\pi_{B,i}^{-1}(1-\pi_{B,i})y_i^2 - N^{-2}\left(\sum_{i\in \mathcal{I}}y_i\right)  \left(\sum_{j\in \mathcal{I}^{c}}y_j\right)
	% \end{eqnarray}
\end{proof}

\section{Bootstrap variance estimator}\label{supp: bootstrap variance estimator}
Since the inclusion probabilities are unavailable for the non-probability sample $A$, we consider a bootstrap variance estimator \citep{kim2019hypothesis} only taking into consideration the design features associated with the probability sample $B$.    

For $b=1,\ldots,B$, let the bootstrap version of (\ref{eq: proposed estimator kl}) be 
\begin{equation}
	\hat{Y}_N^{(b)*} = N^{-1}\sum_{i\in A}\hat{\omega}_i^{(b)*}y_i,\notag
\end{equation}
where $B$ is the number of bootstrap replications,  $ \hat{\omega}_i^{(b)*} = 1+ ( N n_A^{-1}- 1 )\hat{r}_i^{(b)*}$ for $i\in A$, $\hat{\bgamma}^{(b)*} = (\hat{r}_1^{(b)*},\ldots,\hat{r}_N^{(b)*})$ with $\hat{r}_i^{(b)*}=0$ for $i\notin A$ is obtained by 
\begin{eqnarray}
	&\displaystyle\hat{\bgamma}^{(b)*} = \argmin_{\xi_1\leq r_i\leq  \xi_2}\left[\sup_{u\in{\mathcal{H}}}\left\{\frac{S^{(b)*}(\bgamma,u)}{\lVert u\rVert_2^2}-\lambda_1\frac{\lVert u\rVert_{\mathcal{H}}^2}{\lVert u\rVert_2^2}\right\}- \lambda_2 Q_A (\bgamma) \right],\notag\\
	&\displaystyle	S^{(b)*}(\bgamma,u) = \left[ N^{-1}\sum_{i=1}^N\delta_{A,i} \left\{ 1+ \left( \frac{N}{n_A} -1\right)r_i \right\}u(\bx_i) - N^{-1}\sum_{i=1}^N\delta_{B,i} d_{B,i}^{(b)*}u(\bx_i)\right]^2, \notag
\end{eqnarray}
the set of bootstrap weights $\{d_{B,i}^{(b)*}:i\in B\}$ satisfies $$E^*\{N^{-1}\sum_{i=1}^N\delta_{B,i} d_{B,i}^{(b)*}u(\bx_i)\} = N^{-1}\sum_{i=1}^N\delta_{B,i} \pi_{B,i}^{-1}u(\bx_i),$$ $\var^*\{N^{-1}\sum_{i=1}^N\delta_{B,i} d_{B,i}^{(b)*}u(\bx_i)\}=\hat{V}\{ N^{-1}\sum_{i=1}^N\delta_{B,i} \pi_{B,i}^{-1}u(\bx_i)\}$, $E^*(\cdot)$ and $\var^*(\cdot)$ are the conditional expectation and variance with respect to the bootstrap procedure given the probability sample $B$, and $\hat{V}\{ N^{-1}\sum_{i=1}^N\delta_{B,i} \pi_{B,i}^{-1}u(\bx_i)\}$ is a design-unbiased variance estimator of $ N^{-1}\sum_{i=1}^N\delta_{B,i} \pi_{B,i}^{-1}u(\bx_i)$. Then, a bootstrap variance estimator for (\ref{eq: proposed estimator kl}) can be the sample variance of $\{\hat{Y}_N^{(b)*}:b=1,\ldots,B\}$.
The bootstrap variance estimator is reasonable if we can safely ignore the variability with respect to the non-probability sample $A$, for example, when $n_B = o_p(n_A)$ by Assumption~\ref{ass:piB} and Assumption~\ref{ass: prob sample inc prob}

Specifically, if the probability sample $B$ is generated by a Poisson sampling as shown in Section~\ref{sec: simu}, then a design-unbiased variance estimator of $ N^{-1}\sum_{i=1}^N\delta_{B,i} \pi_{B,i}^{-1}u(\bx_i)$ is 
\begin{equation}\label{eq: ss boot var}
	\hat{V}_B = \sum_{i\in B}\frac{1-\pi_{B,i}}{\pi_{B,i}^2}u^2(\bx_i).\notag
\end{equation}
Then, the distribution to generate the bootstrap weights $d_{B,i}^*$ can be normal with mean $\pi_{B,i}^{-1}$ and variance $(1-\pi_{B,i})\pi_{B,i}^{-2}$.

\section{Proof of Theorem~\ref{theorem: CLT}} \label{supp: proof theorem CLT}
\begin{lem}\label{lemma: population size estimator}
	Suppose that Assumption~A\ref{ass: prob sample inc prob} holds. Then, we have 
	\begin{equation*}
		N^{-1}\sum_{i=1}^N\delta_{B,i}\pi_{B,i}^{-1}=O_p(1).
	\end{equation*}
\end{lem}
\begin{proof}[Proof of Lemma~\ref{lemma: population size estimator}]
	Since $P(\delta_{B,i})=\pi_{B,i}$, we have  
	\begin{equation}\label{eq: population mean estimator 1}
		E\left(    N^{-1}\sum_{i=1}^N\delta_{B,i}\pi_{B,i}^{-1}\right) = 1.
	\end{equation}
	Consider 
	\begin{eqnarray}
		\var\left(    N^{-1}\sum_{i=1}^N\delta_{B,i}\pi_{B,i}^{-1}\right)&=& N^{-2}\sum_{i=1}^N\var(\delta_{B,i}\pi_{B,i}^{-1}) + N^{-2}\sum_{i\neq j}\pi_{B,i}^{-1}\pi_{B,j}^{-1}\Cov(\delta_{B,i},\delta_{B,j})\notag \\ 
		&\leq& N^{-2}\sum_{i=1}^N\var(\delta_{B,i}\pi_{B,i}^{-1})\notag \\ 
		&=& N^{-2}\sum_{i=1}^N(1-\pi_{B,i})\pi_{B,i}^{-1}\notag \\ 
		&\leq&N^{-2}C_{B,1}^{-1}N^2n_B^{-1}\notag \\ 
		&=&O(n_B^{-1}),\label{eq: population mean estimator 2}
	\end{eqnarray}
	where the first inequality holds since $\{\delta_{B,i}:i=1,\ldots,N\}$ are negatively associated, and the second inequality holds by Assumption~A\ref{ass: prob sample inc prob}. By Assumption~A\ref{ass: prob sample inc prob}, $n_B^{-1}\to0$, so we have proved Lemma~\ref{lemma: population size estimator} by (\ref{eq: population mean estimator 1})--(\ref{eq: population mean estimator 2}). 
\end{proof}

\begin{lem}\label{lemma: CLT for epsilon}
	Suppose that Assumptions~A\ref{ass: A ind}--A\ref{ass: r bound} hold. Then, we have 
	\begin{equation*}
		B_N^{-1}\sum_{i=1}^N(\delta_{A,i}\hat{w}_i - \delta_{B,i}\pi_{B,i}^{-1})\epsilon_i \to N(0,1),
	\end{equation*}
	where $B_N^2 = \sum_{i=1}^N(\delta_{A,i}\hat{w}_i-\delta_{B,i}\pi_{B,i}^{-1})^2\sigma_i^2$. Besides, $B_N^2 \asymp N^{2}n_B^{-1}$ in probability.
\end{lem}
\begin{proof}[Proof of Lemma~\ref{lemma: CLT for epsilon}]
	Denote $\mathcal{A}_N = \{(\delta_{A,i},\bx_i):i\in A\}\cup \{(\delta_{B,i},\bx_i):i\in B\}$. Then, given $\mathcal{A}_N$, $B_N^2$ is  the conditional variance of $\sum_{i=1}^N(\delta_{A,i}\hat{w}_i - 1)\epsilon_i$. 
	
	We first consider the stochastic order of $B_N$. On the one hand, we have
	\begin{eqnarray}
		B_N^{2} &\leq&C_{\sigma,2}\sum_{i=1}^N(\delta_{A,i}\hat{w}_i-\delta_{B,i}\pi_{B,i}^{-1})^2\notag \\ 
		&\leq& 2C_{\sigma,2} \sum_{i=1}^B\delta_{A,i}\hat{w}_i^2 + 2C_{\sigma,2}\sum_{i=1}^N\delta_{B,i}\pi_{B,i}^{-2}\notag \\ 
		&=& O_p(N^2n_B^{-1}), \label{eq: lemma 9 part 1}
	\end{eqnarray}
	where the last equality holds by Assumption~A\ref{ass: prob sample inc prob}, (\ref{eq: Theorem1 part 21}) and (\ref{eq: w2 subpart 2}).
	On the other hand, consider 
	\begin{eqnarray}
		B_N^{2} &=& \sum_{i=1}^N\delta_{A,i}\hat{w}_i^2\sigma_i^2 -2 \sum_{i=1}^N\delta_{A,i}\delta_{B,i}\hat{w}_i\pi_{B,i}^{-1}\sigma_i^2 +\sum_{i=1}^N\delta_{B,i}\pi_{B,i}^{-2}\sigma_i^2\notag \\ 
		&\geq&C_{\sigma,1}\sum_{i=1}^N\delta_{A,i}\hat{w}_i^2 - 2C_{\sigma,1}\max\{\hat{w}_i:i=1,\ldots,N\}\sum_{i=1}^N\delta_{B,i}\pi_{B,i}^{-1} +C_{\sigma,1}\sum_{i=1}^N\delta_{B,i}\pi_{B,i}^{-2}\notag \\ 
		&\geq&C_{\sigma,1}\sum_{i=1}^N\delta_{A,i}\hat{w}_i^2 -2C_{\sigma,1}\max\{\hat{w}_i:i=1,\ldots,N\}O_p(N) + C_{\sigma,1}C_{B,1}^{-2}N^{2}n_B^{-1}\notag \\ 
		&\geq&C_{\sigma,1} C_{B,1}^{-2}N^{2}n_B^{-1} + o_p(N^{2}n_B^{-1})\label{eq: lemma 9 part 12}
	\end{eqnarray}
	where the second inequality holds by Lemma~\ref{lemma: population size estimator}, and the last inequality holds by the condition that  $\max\{\hat{w}_i:i=1,\ldots,N\}$ is bounded since $\hat{\omega}_i = 1+ ( N n_A^{-1}- 1 )\hat{r}_i$, $\hat{r}_i\leq\xi_2$ and $( N n_A^{-1}- 1 )<C_{A,4}$, where $C_{A,4}$ is discussed in the proof of Theorem~\ref{theorem: convergence rate of proposed estimator}. Thus, by (\ref{eq: lemma 9 part 1})--(\ref{eq: lemma 9 part 12}), we have shown that $B_N^2 \asymp N^{2}n_B^{-1}$ in probability.
	
	For any $\eta>0$, consider 
	\begin{eqnarray}
		&&B_N^{-2}\sum_{i=1}^NE\{\lvert (\delta_{A,i}\hat{w}_i - \delta_{B,i}\pi_{B,i}^{-1})\epsilon_i\rvert^2I\{\lvert(\delta_{A,i}\hat{w}_i - \delta_{B,i}\pi_{B,i}^{-1})\epsilon_i\rvert\geq B_N\eta\}\mid \mathcal{A}_N\}\notag \\ 
		&\leq& \eta^{\delta}B_N^{-2-\delta}\sum_{i=1}^NE\{\lvert (\delta_{A,i}\hat{w}_i - \delta_{B,i}\pi_{B,i}^{-1})\epsilon_i\rvert^{2+\delta}\mid \mathcal{A}_N\}\notag \\ 
		&\leq&\eta^{\delta}B_N^{-2-\delta}C_{\sigma,1}^{-1}C_\delta\max\{\lvert (\delta_{A,i}\hat{w}_i - \delta_{B,i}\pi_{B,i}^{-1})\rvert^{\delta}:i=1,\ldots,N\}B_{N}^2\notag \\ 
		&=&o_p(1),
	\end{eqnarray}
	where the second inequality holds by Assumption~A\ref{ass: epsilon}, and last equality holds since $\max\{\lvert (\delta_{A,i}\hat{w}_i - \delta_{B,i}\pi_{B,i}^{-1})\rvert^{\delta}:i=1,\ldots,N\} = O(N^{\delta}n_B^{-\delta})$ by Assumption~A\ref{ass: prob sample inc prob} and $B_N \asymp Nn_B^{-1/2}$ in probability.
	By a similar argument leading to Theorem~4.1 of \citet{yuan2014conditional}, we have proved Lemma~\ref{lemma: CLT for epsilon}.
	
\end{proof}

\begin{proof}[Proof of Theorem~\ref{theorem: CLT}]
	
	Consider 
	\begin{eqnarray}\label{eq: theorem2 decomposition}
		&&N^{-1} \sum_{i=1}^N(\delta_{A,i}\hat{w}_i-\delta_{B,i}\pi_{B,i}^{-1})y_i - N^{-1}\sum_{i=1}^N(\delta_{A,i}\hat{w}_i -\delta_{B,i}\pi_{B,i}^{-1})\hat{m}(\bx_i)\notag \\ 
		&=&N^{-1} \sum_{i=1}^N(\delta_{A,i}\hat{w}_i-\delta_{B,i}\pi_{B,i}^{-1})\epsilon_i+N^{-1}\sum_{i=1}^N(\delta_{A,i}\hat{w}_i -\delta_{B,i}\pi_{B,i}^{-1})\{m_0(\bx_i) - \hat{m}(\bx_i)\}.\notag\\
	\end{eqnarray}
	
	Lemma~\ref{lemma: CLT for epsilon} validates the central limit theorem for the first part of (\ref{eq: theorem2 decomposition}). By Lemma~\ref{lemma: R S2} and a similar argument in the proof of Lemma~S3 of \citet{wong2018kernel}, we can show that 
	\begin{equation}\label{eq: neg dif}
		N^{-1}\sum_{i=1}^N(\delta_{A,i}\hat{w}_i -\delta_{B,i}\pi_{B,i}^{-1})\{m(\bx_i) - \hat{m}(\bx_i)\}=o_p(n_B^{-1}).
	\end{equation}
	By the stochastic order of $B_N$ in Lemma~\ref{lemma: CLT for epsilon}, we have proved Theorem~\ref{theorem: CLT}.
\end{proof}

\section{Proof of Corollary~\ref{cor: variance estimator}}\label{ss: cor1}
\begin{proof}[Proof of Corollary~\ref{cor: variance estimator}]
	Denote 
	\begin{eqnarray}
		\tilde{\theta} &=& N^{-1}\sum_{i=1}^N\delta_{B,i}\pi_{B,i}^{-1}m(\bx_i) + N^{-1}\sum_{i=1}^N\delta_{A,i}\hat{w}_i\{y_i - m(\bx_i)\}\notag \\ 
		&=& N^{-1}\sum_{i=1}^N\delta_{B,i}\pi_{B,i}^{-1}m(\bx_i) + N^{-1}\sum_{i=1}^N\delta_{A,i}\epsilon_i.\label{eq: theta tilde}
	\end{eqnarray}
	Then, by (\ref{eq: neg dif}), we have 
	\begin{eqnarray}
		\lvert\hat{\theta}_{prop} - \tilde{\theta}\rvert =\left\lvert N^{-1}\sum_{i=1}^N(\delta_{B,i}\pi_{B,i}^{-1} - \delta_{A,i}\hat{w}_i) \{\hat{m}(\bx_i)-m(\bx_i)\}\right\rvert=o_p(n_B^{-1/2}).\label{eq: theta tilde result}
	\end{eqnarray}
	
	Since the asymptotic order of $\hat{\theta}_{prop}$ is $O_p(n_B^{-1/2})$ by Assumption~A\ref{ass: sample B convergence result} and  Theorem~\ref{theorem: CLT},  it is enough to investigate the variance of $\tilde{\theta}$ in (\ref{eq: theta tilde}) by (\ref{eq: theta tilde result}).
	
	Consider 
	\begin{eqnarray}
		\var(\tilde{\theta}) &=& \var\left[ E\left\{N^{-1}\sum_{i=1}^N\delta_{B,i}\pi_{B,i}^{-1}m(\bx_i) + N^{-1}\sum_{i=1}^N\delta_{A,i}\hat{w}_i\epsilon_i\mid \mathcal{A}_N\right\}\right] \notag \\
		&&+ E\left[ \var\left\{N^{-1}\sum_{i=1}^N\delta_{B,i}\pi_{B,i}^{-1}m(\bx_i) + N^{-1}\sum_{i=1}^N\delta_{A,i}\hat{w}_i\epsilon_i\mid \mathcal{A}_N\right\}\right].\notag\\
		\label{eq: tilde theta decomposition}
	\end{eqnarray}
	Since $\epsilon_{i}$ is independent with $\{\delta_{A,i}:i=1,\ldots,N\}$ and $\{\delta_{B,i}:i=1,\ldots,N\}$ by Assumption~A\ref{ass: sample B convergence result}, we have 
	\begin{eqnarray}
		E\left\{N^{-1}\sum_{i=1}^N\delta_{B,i}\pi_{B,i}^{-1}m(\bx_i) + N^{-1}\sum_{i=1}^N\delta_{A,i}\hat{w}_i\epsilon_i\mid \mathcal{A}_N\right\}
		=N^{-1}\sum_{i=1}^N\delta_{B,i}\pi_{B,i}^{-1}m(\bx_i).\notag
	\end{eqnarray}
	Thus, we conclude that  
	\begin{eqnarray}
		&&\var\left[ E\left\{N^{-1}\sum_{i=1}^N\delta_{B,i}\pi_{B,i}^{-1}m(\bx_i) + N^{-1}\sum_{i=1}^N\delta_{A,i}\hat{w}_i\epsilon_i\mid  \mathcal{A}_N\right\}\right]\notag \\
		&=& \var\left\{N^{-1}\sum_{i=1}^N\delta_{B,i}\pi_{B,i}^{-1}m(\bx_i)\right\}.\label{eq: EVVE part 1}
	\end{eqnarray}
	
	Next, consider 
	\begin{eqnarray}
		&&\var\left\{N^{-1}\sum_{i=1}^N\delta_{B,i}\pi_{B,i}^{-1}m(\bx_i) + N^{-1}\sum_{i=1}^N\delta_{A,i}\hat{w}_i\epsilon_i\mid \mathcal{A}_N\right\}\notag \\ 
		&& = N^{-1}\sum_{i=1}^N\hat{w}_i^2\sigma^2_i.\label{eq: EVVE part 2}
	\end{eqnarray}
	Thus, by (\ref{eq: theta tilde result})--(\ref{eq: EVVE part 2}), a plug-in variance estimator of $\hat{\theta}_{prop}$ is 
	\begin{equation}
		\hat{V}\left\{N^{-1}\sum_{i=1}^N\delta_{B,i}\pi_{B,i}^{-1}\hat{m}(\bx_i) \right\} + N^{-2}\sum_{i=1}^N\delta_{A,i}\hat{w}_i^2\{y_i-\hat{m}(\bx_i)\}^2,
	\end{equation}
	where $\hat{V}\left\{N^{-1}\sum_{i=1}^N\delta_{B,i}\pi_{B,i}^{-1}\hat{m}(\bx_i) \right\}$ is a design-based variance of $N^{-1}\sum_{i=1}^N\delta_{B,i}\pi_{B,i}^{-1}\hat{m}(\bx_i)$ treating $\{\hat{m}(\bx_i):\delta_{B,i}=1\}$ as non-stochastic. Thus, we have finished the proof of Corollary~\ref{cor: variance estimator}.
	%, and $\hat{\sigma}^2 = (n_A-1)^{-1}\sum_{i=1}^N\{y_i-\hat{m}(\bx_i)\}^2$
\end{proof}

\section{Doubly robust estimator by \citet{chen2020doubly}}\label{ss: DRE}
Consider a logistic model for the non-probability sample $A$, $\pi_{A,i} = \pi_A(\bx_i;\btheta_0)$, where $\logit\{\pi_A(\bx_i;\btheta_0)\} = \bx_i^{\T}\btheta_0$, and $\btheta_0$ is the true parameter. An estimator of $\btheta_0$, say $\hat{\btheta}$, is obtained by solving 
\begin{equation}
	\sum_{i\in A}\bx_i- \sum_{i\in B}\pi_{B,i}^{-1}\pi_A(\bx_i;\btheta)\bx_i=\bzero. \label{eq: 1}
\end{equation}
The corresponding estimator is termed as ``maximum pseudo-likelihood estimator'' by \citet{chen2020doubly}.

To overcome the model mis-specification for the sampling mechanism associated with the non-probability sample, \citet{chen2020doubly} also proposed  two double robust estimators by assuming a parametric model for $m_0(\bx) = m(\bx;\bbeta_0)$, where $\bbeta_0$ is an unknown parameter to be estimated. Since missing at random is assumed, we can obtain a consistent estimator of $\bbeta_0$, say $\hat{\bbeta}$, using a standard approach. Then, the doubly robust estimators are 
\begin{equation}
	\hat{Y}_1 = N^{-1}\sum_{i\in A}\{\pi_A(\bx_i;\hat{\btheta})\}^{-1}\{y_i-m(\bx_i;\hat{\bbeta})\} + N^{-1}\sum_{i\in B}\pi_{B,i}^{-1}m(\bx_i;\hat{\bbeta})\label{eq: DR 1}
\end{equation}
and
\begin{equation}
	\hat{Y}_2 = \hat{N}^{-1}\sum_{i\in A}\{\pi_A(\bx_i;\hat{\btheta})\}^{-1}\{y_i-m(\bx_i;\hat{\bbeta})\} + \hat{N}^{-1}\sum_{i\in B}\pi_{B,i}^{-1}m(\bx_i;\hat{\bbeta}),\label{eq: DR 2}
\end{equation}
where $\hat{N} = \sum_{i\in A}\{\pi_A(\bx_i;\hat{\btheta})\}^{-1}$.
The only difference between (\ref{eq: DR 1}) and (\ref{eq: DR 2}) is that a true population size $N$ is used for (\ref{eq: DR 1}), but its estimator is used for (\ref{eq: DR 2}).

\section{Additional simulation result}\label{ss: bootstrap variance estimator}
Under a certain simulation setup, denote $\hat{Y}_N^{(m)}$ and $\hat{V}^{(m)}$ to be the HT\_KL estimator and its bootstrap variance estimator for the $m$-th Monte Carlo simulation for $m=1,\ldots,M$, where $M=1\,000$ in the simulation study; see Section~\ref{supp: bootstrap variance estimator} for details about the bootstrap variance estimator. Let $\hat{V}$ be the sample variance of $\{\hat{Y}_N^{(m)}:m=1,\ldots,M\}$. Then, the relative bias of the bootstrap variance estimator is 
$$
\frac{M^{-1}\sum_{m=1}^M\hat{V}^{(m)} - \hat{V}}{\hat{V}}.
$$

Table \ref{tab:my_label rb boot} shows the relative bias of the bootstrap variance estimator for  HT\_KL under different simulation setups. The relative bias of the bootstrap variance estimator is small regardless of the simulation setups. Thus, the variance of HT\_KL can be reasonably estimated by bootstrap.  
\begin{table}
	\centering
	\caption{Relative bias of the bootstrap variance estimator for   HT\_KL  based on 1\,000 Monte Carlo simulations under different simulation setups. The number of replication is $B=200$. }
	\label{tab:my_label rb boot}
	\begin{center}
		\begin{tabular}{ccc}
			\hline 
			Model  &(5\,000, 1\,000,100) & (10\,000, 2\,000,200)\\
			\hline
			Linear&-0.026&0.031\\
			Nonlinear&0.032&0.032\\
			\hline
			
		\end{tabular}
	\end{center}
\end{table}

% \section{BibTeX}

% We hope you've chosen to use BibTeX!\ If you have, please feel free to use the package natbib with any bibliography style you're comfortable with. The .bst file agsm has been included here for your convenience. 

\end{document}